\documentclass[twoside,11pt]{article}

\usepackage[preprint]{jmlr2e}

\usepackage{amsmath}
\usepackage{amsfonts}
\usepackage{tikz}
\usepackage{bbm}
\usepackage{algorithm}
\usepackage{algpseudocode}
\usepackage{chngcntr}
\usepackage{multicol}
\usepackage{multirow}
\usepackage{pdflscape} 
\usepackage{rotating}
\usepackage{tcolorbox}

\counterwithin{algorithm}{section} 

\algnewcommand\algorithmicinput{\textbf{Input:}}
\algnewcommand\Input{\item[\algorithmicinput]}

\algnewcommand\algorithmiciterate{\textbf{Iterate}}
\algnewcommand\Iterate{\item[\algorithmiciterate]}

\algnewcommand\algorithmicdefine{\textbf{Define}}
\algnewcommand\Define{\item[\algorithmicdefine]}

\algnewcommand\algorithmiccompute{\textbf{Compute}}
\algnewcommand\Compute{\item[\algorithmiccompute]}

\algnewcommand\algorithmiccalculate{\textbf{Calculate}}
\algnewcommand\Calculate{\item[\algorithmiccalculate]}

\algnewcommand\algorithmicinitialize{\textbf{Initialize}}
\algnewcommand\Initialize{\item[\algorithmicinitialize]}

\algnewcommand\algorithmicwhileloop{\textbf{While}}
\algnewcommand\Whileloop{\item[\algorithmicwhileloop]}

\algnewcommand\algorithmicendloop{\textbf{Return:}}
\algnewcommand\Endloop{\item[\algorithmicendloop]}

\newtheorem{assumption}[theorem]{Assumption}

\newcommand{\indep}{\perp \!\!\! \perp}
\newcommand{\der}[2]{\dfrac{\textnormal{d}^{#2}#1}{\textnormal{{dt}}^{#2}}}

\DeclareMathOperator{\EX}{\mathbb{E}}
\DeclareMathOperator{\Prob}{\mathbb{P}}
\newcommand{\R}{\mathbb{R}}
\DeclareMathOperator{\expit}{expit}
\DeclareMathOperator*{\argmax}{arg\,max}
\DeclareMathOperator*{\argmin}{arg\,min}
\DeclareMathOperator*{\eqdef }{\stackrel{\text{def}}{=}}

\usepackage{lastpage}
\jmlrheading{00}{2025}{1-\pageref{LastPage}}{6/25; Revised --/--}{--/--}{GROLLEAU25A}{François Grolleau, Céline Beji, Raphaël Porcher and François Petit}

\ShortHeadings{Estimating CACE with Mixtures of Experts}{Grolleau, Beji, Porcher, and Petit}
\firstpageno{1}

\begin{document}

\title{Estimating Complier Average Causal Effects with Mixtures of Experts}

\author{\name François Grolleau \email grolleau@stanford.edu \\
       \addr Stanford Center for Biomedical Informatics Research\\
       Stanford University, Stanford, CA 94305, USA
       \AND
       \name Céline Beji\\
       \name Raphaël Porcher\\
       \name François Petit\\
       \addr Université Paris Cité and Université Sorbonne Paris Nord, Inserm, INRAE,\\ Center for Research in Epidemiology and StatisticS (CRESS), Paris, France}

\editor{Editor to be Assigned}

\maketitle

\begin{abstract}%
Treatment non-compliance, where individuals deviate from their assigned experimental conditions, frequently complicates the estimation of causal effects. To address this, we introduce a novel learning framework based on a mixture of experts architecture to estimate the Complier Average Causal Effect (CACE). Our framework provides a flexible alternative to classical instrumental variable methods by relaxing their strict monotonicity and exclusion restriction assumptions. We develop a principled, two-step procedure where each step is optimized with a dedicated Expectation-Maximization (EM) algorithm. Crucially, we provide formal proofs that the model's components are identifiable, ensuring the learning procedure is well-posed. The resulting CACE estimators are proven to be consistent and asymptotically normal. Extensive simulations demonstrate that our method achieves a substantially lower root mean squared error than traditional instrumental variable approaches when their assumptions fail, an advantage that persists even when our own mixture of experts are misspecified. We illustrate the framework’s practical utility on data from a large-scale randomized trial.
\end{abstract}


\section{Introduction}

Randomized experiments are the gold standard for estimating causal effects, yet their real-world deployments—from clinical trials to large-scale online A/B tests—rarely unfold with textbook simplicity. A frequent complication is non-compliance, where an individual’s assigned treatment (e.g., being encouraged to use a new feature) and their received treatment (e.g., actually using it) diverge for behavioral or logistical reasons \citep{e9r1, prosperi2020}. When this occurs, the overall Average Treatment Effect is often a poor measure of a program's efficacy. While much recent work in machine learning has focused on learning balanced representations to estimate this effect \citep{johansson2016learning}, the more relevant parameter in the presence of non-compliance is the Complier Average Causal Effect (CACE)—the mean effect on the latent subpopulation of individuals who would comply with any assignment given to them \citep{frangakis1999, imbens1994}.

The classic approach to estimating the CACE uses an instrumental variable (IV) framework \citep{angrist1996}. However, its validity rests on two strong and often untestable assumptions: (i) monotonicity, which posits that no individual would actively defy their assignment, and (ii) the exclusion restriction, which assumes that assignment influences the outcome only through the received treatment \citep{stuart2015}. In many modern applications, these assumptions are questionable. For instance, in medical trials, the assignment itself can create placebo effects (violating exclusion restriction), while in online A/B testing, users may actively seek out the experience opposite to that assigned (violating monotonicity) \citep{mansournia2017}. Consequently, settings where both assumptions fail simultaneously are plausible and common. While methods exist to relax one of these assumptions, general frequentist frameworks that relax both while retaining point identification are scarce in the mainstream literature (\emph{see}, e.g., the surveys by \citet{imbens2014instrumental} and \citet{mogstad2024instrumental}).

This paper recasts CACE estimation as a machine learning problem solvable with a mixture of experts model \citep{jacobs1991adaptive, jordan1994, yuksel2012twenty}. We treat the four principal strata—compliers, always-takers, never-takers, and defiers—as unobserved latent classes. Our architecture consists of a gating network that learns to predict an individual's probability of belonging to each latent stratum based on their covariates, and expert networks that model the outcomes for each relevant stratum. This architecture is grounded in the principal ignorability assumption \citep{jo2009}, which posits that an individual's covariates are sufficient to characterize their compliance behavior. By explicitly modeling the latent strata, our framework avoids reliance on the monotonicity and exclusion restriction assumptions. We learn the model parameters using a principled two-step process: we first train the gating network to estimate compliance probabilities, and then use these probabilities to train the expert networks that model outcomes. Each step is optimized using a dedicated Expectation-Maximization (EM) algorithm. Our contributions are as follows:
\begin{enumerate}
\item We introduce a novel mixture of experts framework for CACE estimation that simultaneously relaxes the monotonicity and exclusion restriction assumptions, extending principled causal inference to settings where traditional IV methods are inapplicable.
\item We develop a flexible two-step learning procedure. First, a gating network is trained to estimate compliance probabilities. Second, expert networks are trained to model the conditional potential outcomes from covariates, with the learning process for this step guided by the probabilities from the first. Each step is optimized with a dedicated EM algorithm, and the overall framework yields a suite of four estimators tailored to different combinations of identifying assumptions.
\item We provide robust empirical validation, demonstrating through extensive simulations that our framework achieves a substantially lower Root Mean Squared Error (RMSE) than classical IV approaches when their standard assumptions are violated. This performance advantage persists even when our own mixture of experts are misspecified.
\item We establish formal identifiability guarantees, proving that the core components of our model—the gating and expert networks—are identifiable from observed data under mild, formal technical conditions. This theoretical result ensures our learning procedure is well-posed.
\end{enumerate}

The remainder of this paper is organized as follows. Section \ref{sec:related_work} situates our work in the broader literature. Section \ref{sec:setup} formalizes our causal model. Section \ref{sec:infer} details the estimation procedure. Section \ref{sec:partcases} presents the four specialized estimators derived from our framework. Sections \ref{sec:sim} and \ref{sec:app} present empirical results, and Section \ref{sec:diss} concludes. Reproducibility materials are available at \url{https://github.com/fcgrolleau/CACEmix}.

\section{Related Work}\label{sec:related_work}

Our work is situated within the instrumental variable (IV) framework for estimating the Local Average Treatment Effect (LATE), or CACE, pioneered by \citet{imbens1994} and \citet{angrist1996}. The classical IV approach achieves \emph{point identification} of the treatment effect for compliers but requires four key assumptions: instrument relevance, instrument exogeneity, exclusion restriction, and monotonicity. Our paper focuses on relaxing the latter two, which are often the most difficult to justify in practice.

A significant body of literature has explored this problem. One line of work develops methods to evaluate the plausibility of the core IV assumptions \citep{glymour2012credible,huber2014testing,swanson2015definition,burauel2023evaluating}. Another major branch focuses on \emph{partial identification}, deriving bounds on the CACE when monotonicity or exclusion restriction are violated \citep{balke1997, heckman2001instrumental, flores2013, kitagawa2021}. A third approach seeks to retain point identification by relaxing a single assumption while targeting a specific estimand. This includes methods robust to the failure of monotonicity \citep{de2017tolerating} and machine learning approaches for estimating the average treatment effect when the exclusion restriction fails \citep{sun2022selective}.

Our framework is distinct from these approaches as it aims for point identification while simultaneously relaxing both the monotonicity and exclusion restriction assumptions. We achieve this by replacing these behavioral assumptions with functional form assumptions on the data-generating process—a common strategy in machine learning. This use of covariates to aid identification builds on a long tradition of research in IV settings \citep{abadie2003semiparametric}.
	
Our central identifying condition, principal ignorability (Assumption \ref{as:pi}(ii)), posits that potential outcomes are independent of compliance strata, conditional on covariates. While strong, this assumption enables our framework to identify the distinct components of the mixture-of-experts models without invoking monotonicity or exclusion restriction. Our contribution, therefore, lies in leveraging this modeling trade-off to build a flexible and transparent estimation framework.
	
Finally, the identifiability of our model architecture is grounded in the theory of finite mixture models. This field was pioneered by \citet{teicher1967identifiability}, with key extensions on the identifiability of general mixtures by \citet{jiang1999identifiability}. In Appendix \ref{app:ident}, we provide specific proofs establishing that our proposed gating and expert network structures are identifiable under mild technical conditions, ensuring our learning procedure is well-posed.

\section{Setup}\label{sec:setup}
\subsection{A probability model for the data generating mechanism}
We denote by $X$ the individuals' baseline (i.e., pre-randomization) covariates, and by $Z$ their allocated binary treatment. We formalize the concept of a randomized controlled trial via the following assumption.
\begin{assumption}[{\bf Random allocation}]\label{as:exo} For any realized value of covariates, individuals could be allocated to either treatment option. That is, the allocated treatment is generated as
\begin{align*}
    &Z|X\sim \textit{Bernoulli}(p=\eta(X))
\end{align*}
where $\eta(\cdot)$, the allocation ratio function of the randomized controlled trial\footnote{In many medical RCTs, the allocation ratio is 1:1 independently of $X$ that is, $\eta(\cdot)\equiv0.5$.}, is such that the following holds:
\begin{align*}
&\exists \epsilon_{\eta} \in \R, \quad  \forall x\in \mathcal{X}, \quad 0<\epsilon_{\eta}<\eta(x)<1-\epsilon_{\eta}<1.
\end{align*}
\end{assumption}

To characterize the notions of \emph{complier}, \emph{always taker}, \emph{never taker}, and \emph{defier} individuals we introduce the following potential treatments:
\begin{equation*}
    T^{s=c}\eqdef Z,\quad T^{s=a}\eqdef 1,\quad T^{s=n}\eqdef 0, \;\text{and}\; T^{s=d}\eqdef 1-Z.
\end{equation*}

These potential treatments correspond to the treatment that a \emph{complier}, an \emph{always taker}, a \emph{never taker}, and a \emph{defier} would take respectively. In the notation above, the \emph{s} superscript indicates the ``stratum" of an individual i.e., \emph{complier}, \emph{always taker}, \emph{never taker}, or \emph{defier}. Since the strata are not observed in practice, throughout this paper, we consider it a latent variable. For later convenience, we introduce the latent stratum of an individual as the one-hot-encoded random vector $\boldsymbol{S}=(S_c,S_a,S_n,S_d)^T,$ where the vectors $(1,0,0,0)^T$, $(0,1,0,0)^T$, $(0,0,1,0)^T$, $(0,0,0,1)^T$ indicate that an individual is a \textit{complier}, an \textit{always taker}, a \textit{never taker}, and a \textit{defier} respectively. We further define the conditional probabilities that an individual is a \textit{complier}, an \textit{always taker}, a \textit{never taker}, and a \textit{defier} as $\rho_c(X)\eqdef \EX(S_c|X)$, $\rho_a(X)\eqdef \EX(S_a|X)$, $\rho_n(X)\eqdef \EX(S_n|X)$, and $\rho_d(X)\eqdef \EX(S_d|X)$ respectively. We note that if an individual's stratum and their allocated treatment were known, then the treatment they effectively took would be entirely characterized. To account for this fact, we define the treatment effectively taken $T$ as follows. 
\begin{definition}[{\bf Treatment effectively taken}]\label{def_ttt} The treatment effectively taken $T$ is consistent with potential treatments in the sense that
\begin{align}
    T&\eqdef S_cT^{s=c} + S_aT^{s=a} + S_nT^{s=n}+ S_dT^{s=d} \notag \\
     &=S_c Z + S_a + S_d  (1-Z) \label{eq:ttt}.
\end{align}
\end{definition}
We suppose the existence of elementary potential outcomes of the form $Y^{s=k,z=l,t=m}$ corresponding to the outcome that would be observed if an individual stratum had been $k,$ their allocated treatment $l,$ and their treatment effectively taken $m.$ Although this may appear complex,\footnote{Note that the potential outcome of a \emph{never taker} taking treatment is meaningless i.e., $Y^{s=n,z=0,t=1}$ and $Y^{s=n,z=1,t=1}$ do not have a commonsensical interpretation. For similar reasons, we will not make use of the following elementary potential outcomes: $Y^{s=a,z=1,t=0},$ $Y^{s=a,z=0,t=0},$ $Y^{s=c,z=1,t=0},$ $Y^{s=c,z=0,t=1},$ $Y^{s=d,z=1,t=1},$ and $Y^{s=d,z=0,t=0}$.} these elementary potential outcomes serve the purpose of characterizing the standard potential outcomes while relaxing usual assumptions of exclusion restriction and monotonicity. The standard potential outcomes $Y^{t=1}$ and $Y^{t=0},$ represent the outcome an individual would achieve if they had taken treatment option $T = 0$ or $T = 1$ respectively. In this paper, we define the potential outcomes as follows.
\begin{definition}[{\bf Potential outcomes}]\label{def_po}
The potential outcomes $Y^{t=1}$ and $Y^{t=0}$ are consistent with elementary potential outcomes in the sense that
\begin{align*}
Y^{t=1}&\eqdef S_cY^{s=c,z=1,t=1} + S_aZY^{s=a,z=1,t=1}+ S_a(1-Z)Y^{s=a,z=0,t=1} + S_dY^{s=d,z=0,t=1}, \\
Y^{t=0}&\eqdef S_cY^{s=c,z=0,t=0} + S_nZY^{s=n,z=1,t=0}+ S_n(1-Z)Y^{s=n,z=0,t=0} + S_dY^{s=d,z=1,t=0}.  
\end{align*}
\end{definition}
Note that this definition for $Y^{t=1}$ and $Y^{t=0}$ should not be viewed as a causal assumption as it does not impose any conceptual constraint. In fact, with this definition the potential outcome $Y^{t=1}$ of an \emph{always-taker} could be different depending on whether their allocated treatment was $Z=0$ or $Z=1$; that is, we do not impose $Y^{s=a,z=0,t=1}$ and $ Y^{s=a,z=1,t=1}$ to be equal. Likewise, the potential outcome $Y^{t=0}$ of a \emph{never-taker} could be different depending on whether their allocated treatment was $Z=0$ or $Z=1$. In other words, our model does not make an exclusion restriction assumption. Moreover, in Definition \ref{def_po} the defiers' potential outcomes are explicitly taken into account and hence, our model does not make a monotonicity assumption. In section \ref{sec:partcases}, we will consider the particular situations where the exclusion restriction and/or monotonicity assumptions hold. We introduce the observed outcomes $Y$ by appealing to the standard consistency assumption.
\begin{assumption}[{\bf Consistency}]\label{as:consistency} The observed outcomes are consistent with potential outcomes, in the sense that
\begin{align*}
&Y=TY^{t=1} + (1-T)Y^{t=0}.
\end{align*}
\end{assumption}
To identify complier average causal effects, we also rely on the principal ignorablity and positivity of compliers assumptions.
\begin{assumption}[{\bf Principal ignorablity}]\label{as:pi} The following two conditional independence statements hold.
\begin{itemize}
    \item[(i)] All variables causing the stratum $\boldsymbol{S}$ and the allocated treatment $Z$ are measured, i.e.,
\begin{align*}
&Z\indep \boldsymbol{S}|X. 
\end{align*}
    \item[(ii)] All variables causing the stratum $\boldsymbol{S}$ and the elementary potential outcomes are measured, i.e., $\forall (k,l,m)\in\{c,a,n,d\}\times\{0,1\}\times \{0,1\}$
\begin{align*}
&Y^{s=k,z=l,t=m}\indep \boldsymbol{S}|X. 
\end{align*}
\end{itemize}
\end{assumption}
In a randomized controlled trial, the first conditional independance statement is a weak assumption that often holds by design. The second statement of principal ignorability may be viewed as an adaptation of the usual no unmeasured confounders assumption \citep{rubin1978} to the context of imperfect compliance.
\begin{assumption}[{\bf Positivity of compliers}]\label{as:pos}
There exists a constant $\epsilon_\rho>0$ such that the following holds:
\begin{align*}
&\forall x\in \mathcal{X}, \quad \rho_c(x)> \epsilon_\rho.
\end{align*}
\end{assumption}
Assumption \ref{as:pos} may be viewed as an adaptation of the usual positivity assumption \citep{rosenbaum1983} to the context of imperfect compliance. In practice, assuming that assumptions \ref{as:exo}, \ref{as:pi}, and \ref{as:pos} hold guarantees that all conditional expectations introduced in the next subsection are well-defined.

\subsection{Notations}
For $k\in\{c,a,n,d\}$, $l\in\{0,1\}$ and $m\in\{0,1\},$ we define the conditional probability functions
\begin{align*}
    P_{klm}(x) &\eqdef \EX(S_k|Z=l, T=m, X=x),
\end{align*}
the conditional observed outcome functions
\begin{align*}
    q_{lm}(x) &\eqdef \EX(Y|Z=l,T=m,X=x),
\end{align*}
and the conditional elementary potential outcome functions
\begin{align*}
    Q_{klm}(x) &\eqdef \EX(Y^{s=k,z=l,t=m}|X=x).
\end{align*}
For clarity, we denote the standard propensity score by $e(X)\eqdef \EX[T|X]$ and its relevant adaptation in the context of imperfect compliance is denoted by 
\[
\pi(X,Z)\eqdef \EX(T|X,Z). 
\]
A summary of the notations used in this paper is provided in Appendix \ref{app:notations}.
\subsection{The CACE estimand}
Our target estimand, the CACE, is defined as 
\[
\Delta \eqdef \EX(Y^{t=1}-Y^{t=0}|S_c=1).
\]
In this paper, we make use of the following rearrangement.
\begin{lemma}\label{lemma:est}
Under assumptions \ref{as:pi} and \ref{as:pos}, the CACE estimand can be represented as
\begin{equation*}
    \Delta=\EX\Big[\big\{ Q_{c11}(X)-Q_{c00}(X)\big\}\rho_c(X)\Big] / \EX[\rho_c(X)].
\end{equation*}
\end{lemma}
A proof of this lemma is included in Appendix \ref{app:proofs}. The goal of the inference procedure described in the next section is to estimate the functions $Q_{c11}$, $Q_{c00}$, and $\rho_c$ in order to derive a plug-in estimator for $\Delta$.

\section{Inference}\label{sec:infer}
We consider the experiment $(X_i,\boldsymbol{S_i},Z_i,T_i,Y_i)\stackrel{\text{iid}}\sim \mathcal{P}$ where we only observe $(X_i,Z_i,T_i,Y_i)_{1\le i\le n}$ as $\boldsymbol{S_i}$, the stratum of an individual, is considered a latent variable. An overview of our estimation procedure can be found in Algorithm \ref{metaalgo:1}. Below, we explain the main steps involved.

\subsection{Step 1: joint estimation of the gating network \texorpdfstring{$\rho_k(\cdot)_{k\in\{c,a,n,d\}}$}{rho}}\label{subsec:step1}
First, we will estimate the function $\rho_c$ by making use of the following expression for the propensity score $\pi$.
\begin{lemma}\label{lemma:mix4}
Under assumptions \ref{as:exo} and \ref{as:pi}, the mechanism generating the treatment effectively taken is given by
\begin{align*}
    \pi(X,Z)= \sum_{k\in\{c,a,n,d\}}\rho_k(X)\mu_k(Z)
\end{align*}
where $\mu_k(Z)\eqdef \EX(T^{s=k}|Z).$
\end{lemma}
A proof of this lemma is included in Appendix \ref{app:proofs}. Lemma \ref{lemma:mix4} suggests that the propensity score $\pi(\cdot)$ can be viewed as mixture of the known experts $\mu_c(z)=z$, $\mu_a(z)=1$, $\mu_n(z)=0$, $\mu_d(z)=1-z$, while the proportions of the mixture are given by the unknown gating network $\rho_k(\cdot)_{k\in\{c,a,n,d\}}$. In Theorem \ref{thm:identif_mix1} (Appendix \ref{app:ident}), we show that, under parametric assumptions, the mixture of expert model for $\pi(x,z)$ in Lemma \ref{lemma:mix4} is identifiable. \citet{jordan1994} described procedures to fit mixture of experts. For instance, if the functions  $\rho_k(\cdot)_{k\in\{c,a,n,d\}}$ are assumed to be differentiable with respect to some parameters then, fitting can be achieved in the supervised learning paradigm by specifying the relevant (mixture) architecture and minimizing via gradient descent a binary cross-entropy loss function with targets $(T_i)_{1\leq i\leq n}$.\footnote{This would require training a custom multi-input model where the four known experts $\mu_s(z)_{s\in\{c,a,n,d\}}$ are provided, and the unknown gating network $\rho_k(x;\theta)_{k\in\{c,a,n,d\}}$ is specified as any architecture (differentiable wrt some parameters $\theta$) with output size $(4\times1)$ and softmax activation. Such implementation is feasible in frameworks such as PyTorch \citep{paszke2019pytorch} and Keras \citep{chollet2021deep}.} Alternatively, we propose to jointly estimate $\rho_k(\cdot)_{k\in\{c,a,n,d\}}$ via the procedure given in Algorithm \ref{algo:1}. This procedure details an EM algorithm, based on the description from \citet{Xu1993} for fitting a mixture of known experts. In Algorithm \ref{algo:4}, we provide an adaptation of this EM-procedure that allows to fit nonparametric and/or non-differentiable functions for $\rho_k(\cdot)_{k\in\{c,a,n,d\}}$.

\subsection{Step 2: parallel estimation of the experts \texorpdfstring{$\{Q_{c11},\,Q_{a11}\}$}{Qc11 and Qa11} and \texorpdfstring{$\{Q_{c00},\,Q_{n00}\}$}{QC00, Qn00}}\label{subsec:step2}
Next, we make use of the following rearrangement of the conditional observed outcome functions $q_{11}$ and $q_{00}$ to estimate $Q_{c11}$, and $Q_{c00}$ separately.
 \begin{lemma}\label{lemma:mix}
Suppose that assumptions \ref{as:exo}, \ref{as:consistency}, \ref{as:pi} and \ref{as:pos} hold. Then,
\begin{align*}
    (i) \quad q_{11}(X)&=P_{c11}(X)Q_{c11}(X) + P_{a11}(X)Q_{a11}(X), \\
    (ii) \quad q_{00}(X)&=P_{c00}(X)Q_{c00}(X) + P_{n00}(X)Q_{n00}(X).
\end{align*}
\end{lemma}
A proof of this lemma is included in Appendix \ref{app:proofs}. Furthermore, we can use Bayes’ rule to verify the following lemma.
 \begin{lemma}\label{lemma:P}
Under assumptions \ref{as:exo}, \ref{as:pi} and \ref{as:pos}, the conditional probability functions $P_{c11}$, $P_{a11}$, $P_{c00}$, $P_{n00}$ can be represented as
\begin{align*}
(a) \quad P_{c11}(X)&=\rho_c(X)/\{\rho_c(X)+\rho_a(X)\},\\
(b) \quad P_{a11}(X)&=\rho_a(X)/\{\rho_c(X)+\rho_a(X)\},\\
(c) \quad P_{c00}(X)&=\rho_c(X)/\{\rho_c(X)+\rho_n(X)\},\\
(d) \quad P_{n00}(X)&=\rho_n(X)/\{\rho_c(X)+\rho_n(X)\}.
\end{align*}
\end{lemma}
A proof of this lemma is included in Appendix \ref{app:proofs}. Equation \emph{(i)} in Lemma \ref{lemma:mix} suggests that the conditional expectation $q_{11}(\cdot)$ can be viewed as a mixture of the unknown (expert) functions $Q_{c11}(\cdot)$ and $Q_{a11}(\cdot)$. In addition, in Lemma \ref{lemma:P}, equations \emph{(a)} and \emph{(b)} suggest that the proportions for this mixture, i.e., the gating network $\{P_{c11}(\cdot),P_{a11}(\cdot)\}$, are known, if $\rho_k(\cdot)_{k\in\{c,a,n,d\}}$ are known. Since $\rho_k(\cdot)_{k\in\{c,a,n,d\}}$ are estimated by the end of Step 1, we propose to estimate this gating network via 
\begin{align*}
\hat{P}_{c11}(X_i)&=\hat{\rho}_c(X_i)/\{\hat{\rho}_c(X_i)+\hat{\rho}_a(X_i)\}, \\
\hat{P}_{a11}(X_i)&=1-\hat{P}_{c11}(X_i).&
\end{align*}
In Theorem \ref{thm:identif_mix2_lm} (Appendix \ref{app:ident}), we show that, assuming $Q_{kll}(\cdot)$ are linear parametric functions, the mixtures of the form shown in Lemma \ref{lemma:mix} are identifiable. In Theorem \ref{thm:identif_mix2_expit} (Appendix \ref{app:ident}), we prove under regularity conditions that these mixtures are also identifiable, if we assume that $Q_{kll}(\cdot)$ are expit functions. Adapting the fitting algorithm of \citet{Xu1993}, we propose to jointly estimate $Q_{c11}(\cdot)$ and $Q_{a11}(\cdot)$ via the EM-procedure given in Algorithm \ref{algo:2}, when the outcome $Y$ is binary. This algorithm takes $(X_i,Y_i,\hat{P}_{c11}(X_i))_{i:Z_i=1,T_i=1}$ as input and intuitively, it distinguishes between the compliers and the always takers within the subset $\{Z=1,T=1\}$. Likewise, considering equation \emph{(ii)} from Lemma \ref{lemma:mix}, we propose to jointly estimate $Q_{c00}(\cdot)$ and $Q_{n00}(\cdot)$ by providing $(X_i,Y_i,\hat{P}_{c00}(X_i))_{i:Z_i=0,T_i=0}$ as input to Algorithm \ref{algo:2}. Intuitively, Algorithm \ref{algo:2} then attempts to distinguish between the compliers and the never takers within the subset $\{Z=0,T=0\}$. When the outcome $Y$ is continuous (rather than binary), we propose to use Algorithm \ref{algo:3} (rather than Algorithm \ref{algo:2}), where we fit conditional Gaussian distributions (rather than Binomial distributions) for the experts. Adaptation of these EM-procedures to fit nonparametric experts are given in Algorithm \ref{algo:5} and Algorithm \ref{algo:6} for binary and continuous outcomes respectively.

\subsection{Final step: plug-in estimation of the CACE}
We propose to plug the estimates of the functions $\rho_c,$ $Q_{c11},$ and $Q_{c00}$ into the expression from Lemma \ref{lemma:est} to obtain the following “plug-in/principal-ignorability” estimator for the CACE 
\begin{equation*}
    \widehat{\Delta}_{PI}=\frac{\sum_{i=1}^n \big\{ \hat{Q}_{c11}(X_i)-\hat{Q}_{c00}(X_i)\big\}\hat{\rho}_c(X_i)}{\sum_{i=1}^n\hat{\rho}_c(X_i)}.
\end{equation*}
Assuming generalized linear models for $\rho_k(\cdot)_{k\in\{c,a,n,d\}},$ $Q_{c11}(\cdot),\,Q_{a11}(\cdot)$ and $Q_{c00}(\cdot)\,Q_{n00}(\cdot),$ the estimator $\widehat{\Delta}_{PI}$ jointly solves a set of “stacked” estimating equations. Thus, $\widehat{\Delta}_{PI}$ is a partial M-estimator of $\psi-$type and it follows that under correct parametric model specification, it is $\sqrt{n}$-consistent and asymptotically normal \citep{stefanski2002}. This justifies the use of the bootstrap to estimate the finite sample variance of $\widehat{\Delta}_{PI}$. We provide more details on the stacked estimating equation method in Appendix \ref{app:asymp}.

\begin{algorithm}[H]
\caption{The procedure to estimate the CACE when \hyperref[as:exo]{random allocation}, \hyperref[as:consistency]{consistency}, \hyperref[as:pi]{principal ignorability}, and \hyperref[as:pos]{positivity} assumptions hold.}\label{metaalgo:1}
\begin{algorithmic}
\begin{footnotesize}
\Input Data $(X_i, Z_i, T_i, Y_i)_{1 \leq i \leq n}$.\\
\textbf{\underline{Step 1:}}
\Compute $\hat{\rho}_k(\cdot)_{k\in\{c,a,n,d\}}$ by providing the EM Algorithm \ref{algo:1} with input $(X_i, Z_i, T_i)_{1 \leq i \leq n}.$\\
\textbf{\underline{Step 2:}}
\Calculate estimates for $P_{c11}(X_i)$ and $P_{c00}(X_i)$ as
\begin{align*}
&\hat{P}_{c11}(X_i)=\hat{\rho}_c(X_i)/\{\hat{\rho}_c(X_i)+\hat{\rho}_a(X_i)\}, \\
&\hat{P}_{c00}(X_i)=\hat{\rho}_c(X_i)/\{\hat{\rho}_c(X_i)+\hat{\rho}_n(X_i)\}.&
\end{align*}

\Compute $\hat{Q}_{c11}(\cdot)$ by providing $(X_i, Y_i, \hat{P}_{c11}(X_i))_{i:Z_i=1,T_i=1}$ as input to EM Algorithm \ref{algo:2} if $Y$ is binary, or to EM Algorithm \ref{algo:3} if $Y$ is continuous.

\Compute $\hat{Q}_{c00}(\cdot)$ by providing $(X_i, Y_i, \hat{P}_{c00}(X_i))_{i:Z_i=0,T_i=0}$ as input to EM Algorithm \ref{algo:2} if $Y$ is binary, or to EM Algorithm \ref{algo:3} if $Y$ is continuous.\\
\textbf{\underline{Final step:}}
\Calculate an estimate of $\Delta$ as 
\begin{align*}
    \widehat{\Delta}_{PI}=\frac{\sum_{i=1}^n \big\{ \hat{Q}_{c11}(X_i)-\hat{Q}_{c00}(X_i)\big\}\hat{\rho}_c(X_i)}{\sum_{i=1}^n\hat{\rho}_c(X_i)}.
\end{align*}

\Endloop $\widehat{\Delta}_{PI}$
\end{footnotesize}
\end{algorithmic}
\end{algorithm}

\section{Particular cases where exclusion restriction and/or monotonicity hold}\label{sec:partcases}
In this section, we consider the particular situations where the exclusion restriction and/or monotonicity assumptions hold. We develop specific estimators for CACE that rely on exclusion restriction and/or monotonicity assumptions. For cases where these assumptions hold, our objective was to develop estimators that could enjoy lower (finite-sample) mean squared errors than the estimator $\widehat{\Delta}_{PI}$, and yet remain consistent. An overview of our estimation procedures can be found in Algorithms \ref{metaalgo:2}, \ref{metaalgo:3}, and \ref{metaalgo:4}. Below, we explain the key steps involved.

\subsection{Situations where exclusion restriction holds}
The exclusion restriction assumption can be formalized as follows.
\begin{assumption}[{\bf Exclusion restriction}]\label{as:er}
The allocated treatment is unrelated to potential outcomes for always-takers and never-takers, that is, 
\begin{align*}
Y^{s=a,z=0,t=1} &= Y^{s=a,z=1,t=1} \eqdef Y^{s=a},  \\
Y^{s=n,z=0,t=0} &= Y^{s=n,z=1,t=0} \eqdef Y^{s=n}.
\end{align*}
\end{assumption}
When the above assumption holds, the equations for $Y^{t=1}$ and $Y^{t=0}$ given in Definition \ref{def_po} reduce to:
\begin{align*}
&Y^{t=1} = S_cY^{s=c,z=1,t=1} + S_aY^{s=a} + S_dY^{s=d,z=0,t=1}, \\
&Y^{t=0} = S_cY^{s=c,z=0,t=0} + S_nY^{s=n} + S_dY^{s=d,z=1,t=0}. 
\end{align*}
Conditioning these equations with respect to $T$ and $X$ yields the following result.
\begin{lemma}\label{lemma:mix_er}
Suppose that assumptions \ref{as:exo}, \ref{as:consistency}, \ref{as:pi}, \ref{as:pos} and \ref{as:er} hold. Then,
\begin{align*}
    (i) \quad q_{\cdot 1}(X)&=P_{c\cdot 1}(X)Q_{c11}(X) + P_{a\cdot 1}(X)Q_{a}(X)+ P_{d\cdot1}(X)Q_{d01}(X), \\
    (ii) \quad q_{\cdot 0}(X)&=P_{c\cdot 0}(X)Q_{c00}(X) + P_{n\cdot 0}(X)Q_{n}(X)+ P_{d\cdot 0}(X)Q_{d10}(X)
\end{align*}
where $q_{\cdot m}(X)\eqdef \EX[Y|T=m,X],$ $P_{k\cdot m}(X)\eqdef \EX[S_k|T=m,X],$ and $Q_{k}(X)\eqdef \EX[Y^{s=k}|X].$
\end{lemma}
A proof of this lemma is included in Appendix \ref{app:proofs}. Further, we can use Bayes’ rule and the conditioning of Equation \eqref{eq:ttt} with respect to $X$ to verify the following.
\begin{lemma}\label{lemma:P_er}
Under assumptions \ref{as:exo}, \ref{as:pi}, and \ref{as:pos}, the conditional probabilities $P_{k\cdot m}(X)$ can be represented as 
\begin{align*}
 P_{c\cdot 1}(X) &=\frac{\eta(X)}{e(X)}\rho_c(X), &  P_{c\cdot 0}(X) &=\frac{1-\eta(X)}{1-e(X)}\rho_c(X), \\
    P_{a\cdot 1}(X) &=\frac{\rho_a(X)}{e(X)}, & P_{n\cdot 0}(X) &=\frac{\rho_n(X)}{1-e(X)}, \\
    P_{d\cdot 1}(X) &=\frac{1-\eta(X)}{e(X)}\rho_d(X), &  P_{d\cdot 0}(X) &=\frac{\eta(X)}{1-e(X)}\rho_d(X)
\end{align*}
\noindent where the standard propensity score can be expanded as $$e(X)=\rho_c(X)\eta(X) + \rho_a(X) + \rho_d(X)\{1-\eta(X)\}.$$
\end{lemma}
A proof of this lemma is included in Appendix \ref{app:proofs}. In consequence of these results, when exclusion restriction holds, we propose an adaptation of the estimation procedure for step 2 (section \ref{subsec:step2}). 

Lemma \ref{lemma:P_er}, suggests straightforward plug-in estimators for $\{P_{k\cdot m}(X_i)\}_{1\leq i\leq n}.$  As such, equation \emph{(i)} from Lemma \ref{lemma:mix_er}, 
suggests to jointly estimate $Q_{c11}(\cdot),$ $Q_{a}(\cdot),$ and $Q_{d01}(\cdot),$ via a procedure able to fit a mixture of three experts where the proportions of the mixture are already known. We propose to do so via the EM-procedure given in Algorithm \ref{algo:7}, when the outcome $Y$ is binary. This algorithm takes $(X_i,Y_i,\hat{P}_{c\cdot1}(X_i),\hat{P}_{a\cdot1}(X_i),\hat{P}_{d\cdot1}(X_i))_{i:T_i=1}$ as input and intuitively, it distinguishes between the compliers, the always takers and the defiers within the subset $\{T=1\}$. 

Likewise, considering equation \emph{(ii)} from Lemma \ref{lemma:mix_er}, we propose to jointly estimate $Q_{c00}(\cdot),$ $Q_{n}(\cdot)$ and $Q_{d10}(\cdot)$ by providing $(X_i,Y_i,\hat{P}_{c\cdot0}(X_i),\hat{P}_{n\cdot0}(X_i), \hat{P}_{d\cdot0}(X_i))_{i:T_i=0}$ as input to Algorithm \ref{algo:7}. Intuitively, Algorithm \ref{algo:7} then attempts to distinguish between the compliers, the never takers and the defiers within the subset $\{T=0\}$. When the outcome $Y$ is continuous (rather than binary), we propose to use Algorithm \ref{algo:8} (rather than Algorithm \ref{algo:7}). Adaptation of these EM-procedures to fit nonparametric experts are given in Algorithm \ref{algo:9} and Algorithm \ref{algo:10} for binary and continuous outcomes respectively.
\begin{algorithm}
\caption{The procedure to estimate the CACE when \hyperref[as:exo]{random allocation}, \hyperref[as:consistency]{consistency}, \hyperref[as:pi]{principal ignorability}, \hyperref[as:pos]{positivity}, and \hyperref[as:er]{exclusion restriction} assumptions hold.}\label{metaalgo:2}
\begin{algorithmic}
\begin{footnotesize}
\Input Data $(X_i, Z_i, T_i, Y_i)_{1 \leq i \leq n}$.\\
\textbf{\underline{Step 1:}}
\Compute $\hat{\rho}_k(\cdot)_{k\in\{c,a,n,d\}}$ by providing the EM Algorithm \ref{algo:1} with input $(X_i, Z_i, T_i)_{1 \leq i \leq n}.$\\
\textbf{\underline{Step 2:}}
\Calculate estimates for $P_{c\cdot 1}(X_i)$, $P_{a\cdot 1}(X_i)$, $P_{d\cdot 1}(X_i)$, $P_{c\cdot 0}(X_i)$, $P_{n\cdot 0}(X_i)$, and  $P_{d\cdot 0}(X_i)$ as
\begin{align*}
\hat{P}_{c\cdot 1}(X_i) \! &= \! \frac{\hat{\eta}(X_i)}{\hat{e}(X_i)}\hat{\rho}_c(X_i), 
&\hat{P}_{c\cdot 0}(X_i) \!  &= \! \frac{1 \!- \! \hat{\eta}(X_i)}{ 1 \! - \! \hat{e}(X_i)}\hat{\rho}_c(X_i), \\
\hat{P}_{a\cdot 1}(X_i) &=\frac{\hat{\rho}_a(X_i)}{\hat{e}(X_i)},
&\hat{P}_{n\cdot 0}(X_i) &=\frac{\hat{\rho}_n(X_i)}{1-\hat{e}(X_i)}, \\
\hat{P}_{d\cdot 1}(X_i) &=\frac{1-\hat{\eta}(X_i)}{\hat{e}(X_i)}\hat{\rho}_d(X_i),
& \hat{P}_{d\cdot 0}(X_i) &=\frac{\hat{\eta}(X_i)}{1-\hat{e}(X_i)}\hat{\rho}_d(X_i)
\end{align*}
where 
\begin{align*}
\hat{e}(X_i)&=\hat{\rho}_c(X_i)\hat{\eta}(X_i) + \hat{\rho}_a(X_i) + \hat{\rho}_d(X_i)\{1-\hat{\eta}(X_i)\}\text{ and}&\\ \hat{\eta}(X_i)&=\hat{\EX}[Z|X_i].
\end{align*}

\Compute $\hat{Q}_{c11}(\cdot)$ by providing: 
\begin{equation*}
(X_i, Y_i, \hat{P}_{c\cdot 1}(X_i), \hat{P}_{a\cdot 1}(X_i), \hat{P}_{d\cdot 1}(X_i))_{i:T_i=1} 
\end{equation*}
as input to EM Algorithm \ref{algo:7} if $Y$ is binary, or to EM Algorithm \ref{algo:8} if $Y$ is continuous.

\Compute $\hat{Q}_{c00}(\cdot)$ by providing:
\begin{equation*}
(X_i, Y_i, \hat{P}_{c\cdot 0}(X_i), \hat{P}_{n\cdot 0}(X_i), \hat{P}_{d\cdot 0}(X_i))_{i:T_i=0}
\end{equation*}
as input to EM Algorithm \ref{algo:7} if $Y$ is binary, or to EM Algorithm \ref{algo:8} if $Y$ is continuous.\\
\textbf{\underline{Final step:}}
\Calculate an estimate of $\Delta$ as 
\begin{align*}
    \widehat{\Delta}_{PI}^{ER}=\frac{\sum_{i=1}^n \big\{ \hat{Q}_{c11}(X_i)-\hat{Q}_{c00}(X_i)\big\}\hat{\rho}_c(X_i)}{\sum_{i=1}^n\hat{\rho}_c(X_i)}.
\end{align*}

\Endloop $\widehat{\Delta}_{PI}^{ER}$
\end{footnotesize}
\end{algorithmic}
\end{algorithm}

\subsection{Situations where monotonicity holds}
The monotonicity assumption can be expressed as follows.
\begin{assumption}[{\bf Monotonicity}]\label{as:monotonicity} There are no defiers in the population, i.e.,
\begin{align*}
&\rho_d(\cdot)\equiv0.
\end{align*}
\end{assumption}
When the above assumption holds, the relation in Lemma \ref{lemma:mix4} simplifies as follows. 
\begin{lemma}\label{lemma:monotonicity}
Under assumptions \ref{as:exo}, \ref{as:pi}, and \ref{as:monotonicity}, the mechanism generating the treatment effectively taken is given by
\begin{align*}
    \pi(X,Z)= \sum_{k\in\{c,a,n\}}\rho_k(X)\mu_k(Z).
\end{align*}
\end{lemma}
\begin{proof}
    Follows directly from Lemma \ref{lemma:mix4} and Assumption \ref{as:monotonicity}.
\end{proof}
Building on Lemma \ref{lemma:monotonicity}, when the monotonicity assumption holds, we propose an adaptation of step 1 (section \ref{subsec:step1}) and propose to jointly estimate $\rho_k(\cdot)_{k\in\{c,a,n\}}$ via the procedure given in Algorithm \ref{algo:11}. In Algorithm \ref{algo:12}, we adapt this EM-procedure to fit nonparametric functions for $\rho_k(\cdot)_{k\in\{c,a,n\}}$. 
\begin{algorithm}[H]
\caption{The procedure to estimate the CACE when \hyperref[as:exo]{random allocation}, \hyperref[as:consistency]{consistency}, \hyperref[as:pi]{principal ignorability}, \hyperref[as:pos]{positivity}, and \hyperref[as:monotonicity]{monotonicity} assumptions hold.}\label{metaalgo:3}
\begin{algorithmic}
\begin{footnotesize}
\Input Data $(X_i, Z_i, T_i, Y_i)_{1 \leq i \leq n}$.\\
\textbf{\underline{Step 1:}}
\Compute $\hat{\rho}_k(\cdot)_{k\in\{c,a,n\}}$ by providing the EM Algorithm \ref{algo:11} with input $(X_i, Z_i, T_i)_{1 \leq i \leq n}.$\\
\textbf{\underline{Step 2:}}
\Calculate estimates for $P_{c11}(X_i)$ and $P_{c00}(X_i)$ as
\begin{align*}
\hat{P}_{c11}(X_i)&=\hat{\rho}_c(X_i)/\{\hat{\rho}_c(X_i)+\hat{\rho}_a(X_i)\},\\ \hat{P}_{c00}(X_i)&=\hat{\rho}_c(X_i)/\{\hat{\rho}_c(X_i)+\hat{\rho}_n(X_i)\}.&
\end{align*}

\Compute $\hat{Q}_{c11}(\cdot)$ by providing $(X_i, Y_i, \hat{P}_{c11}(X_i))_{i:Z_i=1,T_i=1}$ as input to EM Algorithm \ref{algo:2} if $Y$ is binary, or to EM Algorithm \ref{algo:3} if $Y$ is continuous.

\Compute $\hat{Q}_{c00}(\cdot)$ by providing $(X_i, Y_i, \hat{P}_{c00}(X_i))_{i:Z_i=0,T_i=0}$ as input to EM Algorithm \ref{algo:2} if $Y$ is binary, or to EM Algorithm \ref{algo:3} if $Y$ is continuous.\\
\textbf{\underline{Final step:}}
\Calculate an estimate of $\Delta$ as 
\begin{align*}
    \widehat{\Delta}_{PI}^{MO}=\frac{\sum_{i=1}^n \big\{ \hat{Q}_{c11}(X_i)-\hat{Q}_{c00}(X_i)\big\}\hat{\rho}_c(X_i)}{\sum_{i=1}^n\hat{\rho}_c(X_i)}.
\end{align*}

\Endloop $\widehat{\Delta}_{PI}^{MO}$
\end{footnotesize}
\end{algorithmic}
\end{algorithm}

\subsection{Situations where monotonicity and exclusion restriction hold}
When both monotonicity and exclusion restriction assumptions hold, Lemmas \ref{lemma:mix_er} and \ref{lemma:P_er} straightforwardly simplify as follows.
\begin{lemma}\label{lemma:P_er_mon}
Under assumption \ref{as:exo}, \ref{as:pi}, \ref{as:pos}, and \ref{as:monotonicity}, the conditional probabilities $P_{k\cdot m}(X)$ can be represented as
\begin{align*}
    P_{c\cdot 1}(X) &=\frac{\eta(X)}{e(X)}\rho_c(X), &P_{c\cdot 0}(X) &=\frac{1-\eta(X)}{1-e(X)}\rho_c(X), \\
    P_{a\cdot 1}(X) &=\frac{\rho_a(X)}{e(X)}, &P_{n\cdot 0}(X) &=\frac{\rho_n(X)}{1-e(X)}, \\
    P_{d\cdot 1}(X) &=0, &P_{d\cdot 0}(X) &=0
\end{align*}
where the standard propensity score can be expanded as $$e(X)=\rho_c(X)\eta(X) + \rho_a(X).$$
\begin{proof}
    Follows directly from Lemma \ref{lemma:P_er} and Assumption \ref{as:monotonicity}.
\end{proof}
\end{lemma}
\begin{lemma}\label{lemma:mix_er_mon}
Under assumptions \ref{as:exo}, \ref{as:consistency}, \ref{as:pi}, \ref{as:pos}, \ref{as:er} and \ref{as:monotonicity}, the following relations hold
\begin{align*}
    (i) \quad q_{\cdot 1}(X)&=P_{c\cdot 1}(X)Q_{c11}(X) + P_{a\cdot 1}(X)Q_{a}(X), \\
    (ii) \quad q_{\cdot 0}(X)&=P_{c\cdot 0}(X)Q_{c00}(X) + P_{n\cdot 0}(X)Q_{n}(X).
\end{align*}
\begin{proof}
    Follows directly from Lemma \ref{lemma:mix_er} and Lemma \ref{lemma:P_er_mon}.
\end{proof}
\end{lemma}
In Lemma \ref{lemma:mix_er_mon}, Equation \emph{(i)} suggests to jointly estimate $Q_{c11}(\cdot)$ and $Q_{a}(\cdot)$ via a procedure able to fit a mixture of two experts where the proportions of the mixture are already known. Accordingly, when the outcome $Y$ is binary, we propose to provide the following inputs $(X_i,Y_i,\hat{P}_{c\cdot1}(X_i))_{i:T_i=1}$ to an EM Algorithm such as \ref{algo:2}. Likewise, to jointly estimate $Q_{c00}(\cdot)$ and $Q_{n}(\cdot),$ inputs $(X_i,Y_i,\hat{P}_{c\cdot0}(X_i))_{i:T_i=0}$ should be provided. When the outcome $Y$ is continuous, an EM-procedure such as Algorithm \ref{algo:3} can be used. Nonparametric alternatives to these procedures are available as Algorithm \ref{algo:5} and \ref{algo:6} respectively.

\begin{algorithm}
\caption{The procedure to estimate the CACE when \hyperref[as:exo]{random allocation}, \hyperref[as:consistency]{consistency}, \hyperref[as:pi]{principal ignorability}, \hyperref[as:pos]{positivity}, \hyperref[as:er]{exclusion restriction}, and \hyperref[as:monotonicity]{monotonicity} assumptions hold.}\label{metaalgo:4}
\begin{algorithmic}
\begin{footnotesize}
\Input Data $(X_i, Z_i, T_i, Y_i)_{1 \leq i \leq n}$.\\
\textbf{\underline{Step 1:}}
\Compute $\hat{\rho}_k(\cdot)_{k\in\{c,a,n\}}$ by providing the EM Algorithm \ref{algo:11} with input $(X_i, Z_i, T_i)_{1 \leq i \leq n}.$\\
\textbf{\underline{Step 2:}}
\Calculate estimates for $P_{c\cdot 1}(X_i)$ and $P_{c\cdot 0}(X_i)$ as
\begin{align*}
\hat{P}_{c\cdot 1}(X_i) &=\frac{\hat{\eta}(X_i)}{\hat{e}(X_i)}\hat{\rho}_c(X_i), 
&\hat{P}_{c\cdot 0}(X_i) &=\frac{1-\hat{\eta}(X_i)}{1-\hat{e}(X_i)}\hat{\rho}_c(X_i)
\end{align*}
where 
\begin{align*}
&\hat{e}(X_i)=\hat{\rho}_c(X_i)\hat{\eta}(X_i) + \hat{\rho}_a(X_i)&\text{and}&&\hat{\eta}(X_i)=\hat{\EX}[Z|X_i].&
\end{align*}

\Compute $\hat{Q}_{c11}(\cdot)$ by providing $(X_i, Y_i, \hat{P}_{c\cdot 1}(X_i))_{i:T_i=1}$ as input to an EM Algorithm such as \ref{algo:2} if $Y$ is binary, or \ref{algo:3} if $Y$ is continuous.

\Compute $\hat{Q}_{c00}(\cdot)$ by providing $(X_i, Y_i, \hat{P}_{c\cdot 0}(X_i))_{i:T_i=0}$ as input to an EM Algorithm such as \ref{algo:2} if $Y$ is binary, or \ref{algo:3} if $Y$ is continuous.\\
\textbf{\underline{Final step:}}
\Calculate an estimate of $\Delta$ as 
\begin{align*}
    \widehat{\Delta}_{PI,MO}^{ER}=\frac{\sum_{i=1}^n \big\{ \hat{Q}_{c11}(X_i)-\hat{Q}_{c00}(X_i)\big\}\hat{\rho}_c(X_i)}{\sum_{i=1}^n\hat{\rho}_c(X_i)}.
\end{align*}

\Endloop $\widehat{\Delta}_{PI,MO}^{ER}$
\end{footnotesize}
\end{algorithmic}
\end{algorithm}

\section{Simulations}\label{sec:sim}
\subsection{Description}
We conduct a simulation study to evaluate the finite sample properties of the CACE estimation methodologies detailed in Algorithms \ref{metaalgo:1}, \ref{metaalgo:2}, \ref{metaalgo:3} and \ref{metaalgo:4}. To this end, we generate a target population of 10 million individuals from which we draw 1000 random samples of size $n = 2000,4000$ and $10000.$ The datasets comprise $14$ correlated covariates: $7$ Bernoulli distributed and $7$ log-normally distributed. We vary the data-generating mechanism to consider all four situations where exclusion restriction and monotonicity assumptions either do or do not hold. We also consider scenarios where the required parametric models are either well specified or misspecified as a result of two relevant variables being omitted. The full description of our data-generating mechanism is provided in Appendix \ref{app:sim_desc}. We compare our estimation methodologies to two estimators from the instrumental variable literature i.e., the standard Wald estimator \citep{angrist1996, wald1940} and the IV matching estimator \citep[equation 14]{frolich2007} (\emph{see} Appendix \ref{app:sim_iv}). In total, we examine 144 different combinations of data generating scenario $\times$ estimator $\times$ specification choice $\times$ sample size. 

\subsection{Results}
\begin{table*}
\caption{Results of the simulation study for misspecified parametric models. \label{tab:sim_res}}
\centering
\scalebox{.55}{
\begin{tabular}{lcccccccccccc}
\hline
 Assumptions  &   \multicolumn{3}{c}{Scenario 1}  &  \multicolumn{3}{c}{Scenario 2}  &  \multicolumn{3}{c}{Scenario 3}  & \multicolumn{3}{c}{Scenario 4} \\ 
Principal ignorability &  \multicolumn{3}{c}{$+$}  &  \multicolumn{3}{c}{$+$}  &  \multicolumn{3}{c}{$+$}  &  \multicolumn{3}{c}{$+$}  \\
Exclusion restriction &  \multicolumn{3}{c}{$-$}  &  \multicolumn{3}{c}{$+$}  &  \multicolumn{3}{c}{$-$}  &  \multicolumn{3}{c}{$+$}  \\
Monotonicity &  \multicolumn{3}{c}{$-$}  &  \multicolumn{3}{c}{$-$}  &  \multicolumn{3}{c}{$+$}  &  \multicolumn{3}{c}{$+$}  \\ \hline
Estimator & Bias (\%) & SE (\%) & RMSE  (\%) & Bias (\%) & SE (\%) & RMSE  (\%) & Bias (\%) & SE (\%) & RMSE  (\%) & Bias (\%) & SE (\%) & RMSE  (\%) \\ \hline 
$n=2\,000$ & & & & & & & & & & & &  \\ 
$\widehat{\Delta}_{PI}$ & $-$3.27 & {\bf 4.88} & {\bf 5.88} & $-$3.97 & {\bf 4.01} & {\bf 5.64} & $-$3.46 &  4.13 &  5.39 & $-$4.06 & 3.19 & 5.16 \\ 
$\widehat{\Delta}_{PI}^{ER}$ & $-$7.63 & 6.51 & 10.03 & $-$7.93 & 6.12 & 10.02 & $-$3.62 & 4.88 & 5.39 & $-$4.06 & 3.19 & 5.16 \\ 
$\widehat{\Delta}_{PI,MO}$ & {\bf $-$1.17} & 7.04 & 7.14 & {\bf $-$3.09} & 5.86 & 6.63 & $-$2.67 & 4.35 & {\bf 5.10} & $-$3.13 & 3.38 & 4.60 \\ 
$\widehat{\Delta}_{PI,MO}^{ER}$ & $-$4.97 & 8.94 & 10.23 & $-$3.37 & 8.12 & 8.80 & {\bf $-$2.12} & 4.92 & 5.36 & $-$2.25 & 4.01 & 4.59 \\ 
$\widehat{\Delta}_{IV\,\text{matching}}$ & 57.57 &  7.17 & 58.01 & 39.55 & 6.56 & 40.09 & 10.78 & 3.46 & 11.32 & 0.13 & 3.10 & 3.10 \\ 
$\widehat{\Delta}_{IV\,\text{Wald}}$ & 57.52 & 7.09 & 57.96 & 39.53 & 6.53 & 40.07 & 10.76 & {\bf 3.41} & 11.29 & {\bf 0.12} & {\bf 3.08} & {\bf 3.09} \\ 

$n=5\,000$ & & & & & & & & & &  \\ 
$\widehat{\Delta}_{PI}$ & -2.18 & {\bf 3.38} & {\bf 4.02} & $-$2.64 & 2.58 & 3.69 & $-$2.68 & 2.78 & 3.86 & -3.12 & 2.14 & 3.78\\ 
$\widehat{\Delta}_{PI}^{ER}$ & $-$5.71 & 4.14 & 7.05 & $-$5.00 & 3.82 & 6.29 & $-$2.89 & 2.97 & 4.14 & $-$3.00 & 2.38 & 3.83 \\ 
$\widehat{\Delta}_{PI,MO}$ & {\bf $-$0.73} & 4.60 & 4.65 & $-$1.82  & 3.79 & 4.21 & {\bf $-$1.71} & 2.88 & {\bf 3.35} & $-$2.25 &  2.21 & 3.16 \\ 
$\widehat{\Delta}_{PI,MO}^{ER}$ & $-$3.34 & 5.30 & 6.27 & {\bf $-$1.64} & 4.47 & 4.77 & $-$2.00 & 3.02 & 3.62 & $-$1.76 & 2.38 & 2.97 \\ 
$\widehat{\Delta}_{IV\,\text{matching}}$ & 57.29 & 4.59 & 57.47 & 39.36 & 4.11 & 39.57 & 10.64 & 2.12 & 10.84 & {\bf 0.02} & 1.98 & 1.98 \\ 
$\widehat{\Delta}_{IV\,\text{Wald}}$ & 57.28 & 4.56 & 57.46 & 39.36 & 4.08 & 39.57 & 10.64 & {\bf 2.11} & 10.85 & 0.03 & {\bf 1.97} &  {\bf 1.97}  \\ 

$n=10\,000$ & & & & & & & & & & & &  \\ 
$\widehat{\Delta}_{PI}$ & $-$1.78 & {\bf 2.51} & {\bf 3.08} & $-$2.42 & {\bf 1.97} & {\bf 3.12} & $-$2.58 &  1.99 &  3.25 & $-$3.06 & 1.55 & 3.43 \\ 
$\widehat{\Delta}_{PI}^{ER}$ & $-$5.31 & 3.09 & 6.14 & $-$4.52 & 2.59 & 5.21 & $-$2.92 & 2.08 & 3.59 & $-$3.05 & 1.77 & 3.52 \\ 
$\widehat{\Delta}_{PI,MO}$ & {\bf $-$0.91} & 3.32 & 3.44 & {\bf $-$1.92} & 2.49 & 3.15 & {\bf  $-$1.60} & 2.05 & {\bf 2.61} & $-$2.15 & 1.61 & 2.68 \\ 
$\widehat{\Delta}_{PI,MO}^{ER}$ & $-$2.77 & 3.57 & 4.52 & $-$1.27 & 3.04 & 3.30 & $-$1.96 & 2.07 & 2.85 & $-$1.61 & 1.65 & 2.30 \\ 
$\widehat{\Delta}_{IV\,\text{matching}}$ & 57.21 &  3.27 & 57.31 & 39.25 & 2.94 & 39.36 & 10.58 & 1.54 & 10.70  & $-$0.02 & 1.40 & 1.40 \\ 
$\widehat{\Delta}_{IV\,\text{Wald}}$ & 57.22 & 3.26 & 57.32 & 39.26 & 2.94 & 39.37 & 10.59& {\bf 1.53} & 10.70 & {\bf $-$0.02} & {\bf 1.40} & {\bf 1.40} \\ 
\hline
\end{tabular}
}
{\small {\raggedright Scenario 1, $\EX[Y^{t=1}-Y^{t=0}]= -2.13\% $,  $\Delta= 20.31\% $, $\EX[(Y^{s=a,z=0,t=1}-Y^{s=a,z=1,t=1})^2]= 60.74\% $, $\EX[(Y^{s=n,z=1,t=0}-Y^{s=a,z=0,t=0})^2]= 54.90\% $, $\EX[\rho_c(X)]= 55.08\% $, and $\EX[\rho_d(X)]= 14.36\% $.  \par}
{\raggedright Scenario 2, $\EX[Y^{t=1}-Y^{t=0}]= -5.79\% $,  $\Delta= 20.31\% $, $\EX[(Y^{s=a,z=0,t=1}-Y^{s=a,z=1,t=1})^2]= 0 $, $\EX[(Y^{s=n,z=1,t=0}-Y^{s=a,z=0,t=0})^2]= 0 $, $\EX[\rho_c(X)]= 55.08\% $, and $\EX[\rho_d(X)]= 14.36\% $.  \par}
{\raggedright Scenario 3, $\EX[Y^{t=1}-Y^{t=0}]= 13.54\% $,  $\Delta= 20.31\% $, $\EX[(Y^{s=a,z=0,t=1}-Y^{s=a,z=1,t=1})^2]= 60.74\% $, $\EX[(Y^{s=n,z=1,t=0}-Y^{s=a,z=0,t=0})^2]= 54.90\% $, $\EX[\rho_c(X)]= 65.30\% $, and $\EX[\rho_d(X)]=0 $. \par}
{\raggedright  Scenario 4, $\EX[Y^{t=1}-Y^{t=0}]= 10.08\% $,  $\Delta= 20.31\% $, $\EX[(Y^{s=a,z=0,t=1}-Y^{s=a,z=1,t=1})^2]= 0$, $\EX[(Y^{s=n,z=1,t=0}-Y^{s=a,z=0,t=0})^2]= 0 $, $\EX[\rho_c(X)]= 56.30\% $, and $\EX[\rho_d(X)]= 0 $.  \par}}

\end{table*}

Results across multiple data-generating scenarios highlight the estimators’ relative strengths and disadvantages (Figure \ref{fig:sim_res}). In terms of Root Mean Squared Error (RMSE), under misspecified parametric models, no estimator achieved the best performance across all data-generating scenarios (Table \ref{tab:sim_res}).
\begin{figure}
\centering
\scalebox{.55}{
\includegraphics{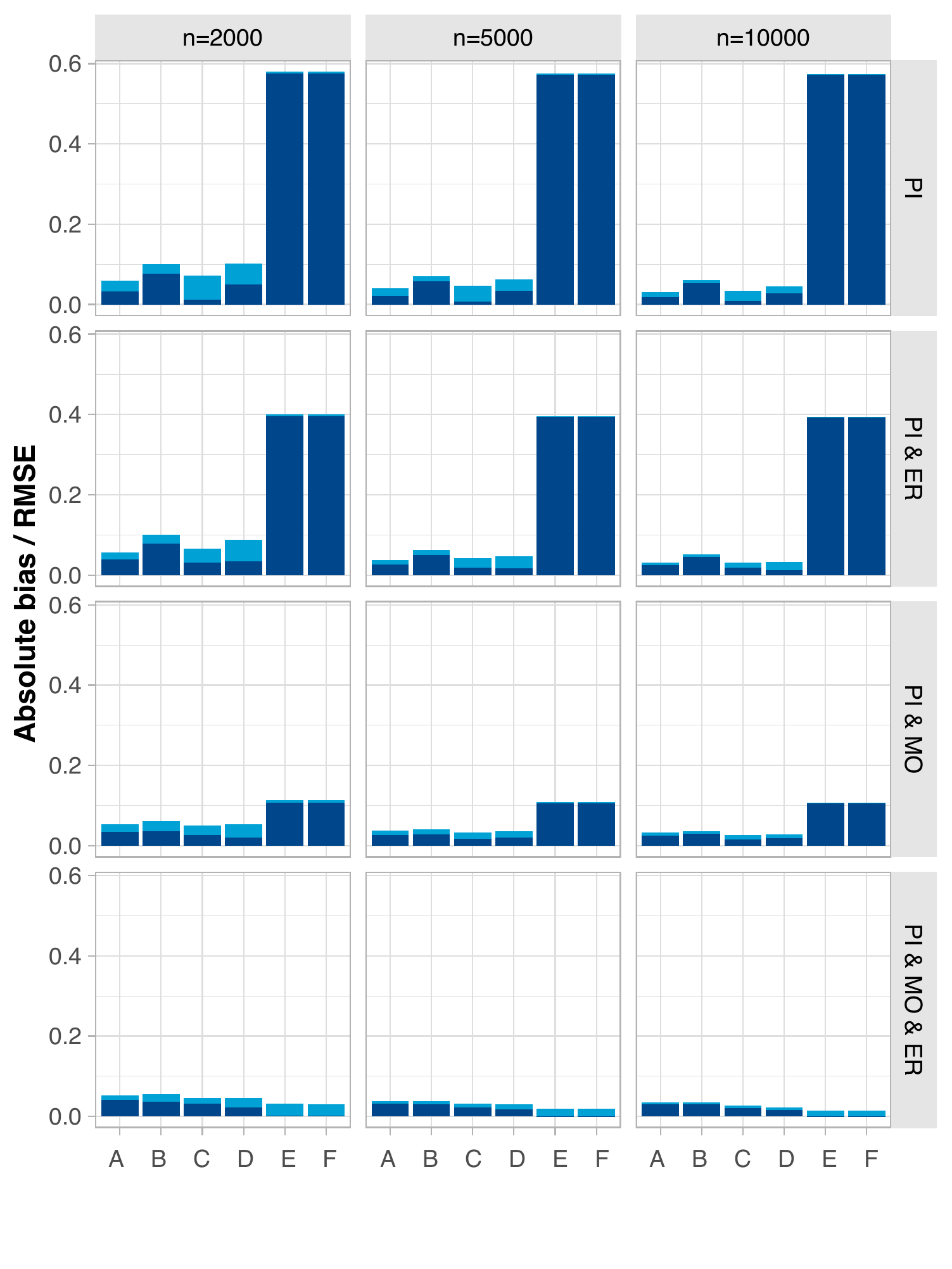}}
\caption{Estimators' absolute bias and Root Mean Squared Error (RMSE) under misspecified parametric models across twelve scenario/sample size combinations.}
\label{fig:sim_res}
\scriptsize  Absolute bias is the darker portion of each bar; RMSE corresponds to the total bar size. Letters A, B, C, D, E and F indicate the estimators $\widehat{\Delta}_{PI}$, $\widehat{\Delta}_{PI}^{ER}$, $\widehat{\Delta}_{PI,MO}$, $\widehat{\Delta}_{PI,MO}^{ER}$, $\widehat{\Delta}_{IV matching}$, and $\widehat{\Delta}_{IV wald}$ respectively. Abbreviations: PI = Principal Ignorability (Scenario 1), PI \& ER = Principal Ignorability and Exclusion Restriction (Scenario 2), PI \& MO= Principal Ignorability and Monotonicity (Scenario 3), PI \& ER \& MO = Principal Ignorability, Exclusion Restriction and Monotonicity (Scenario 4).
\end{figure}
However, as anticipated, the estimator $\widehat{\Delta}_{PI}$ achieved the lowest RMSE in the scenario where neither monotonicity nor exclusion restriction holds (Scenario 1), while in the scenario where monotonicity holds but exclusion restriction does not (Scenario 3), the estimator $\widehat{\Delta}_{PI,MO}$ did. More surprisingly, in the scenario where exclusion restriction holds but monotonicity does not, the estimator $\widehat{\Delta}_{PI}^{ER}$ achieved the worst RMSE among our proposed estimation methodologies. This may be due to the fact that the requirement for $\widehat{\Delta}_{PI}^{ER}$ to fit mixtures of three experts (as in Lemma \ref{lemma:mix_er}) rather than two (as in Lemma \ref{lemma:mix}) was not compensated by the use of more data.\footnote{Note that in Lemma \ref{lemma:mix_er}, the conditioning is on $\{T=1\}$ or $\{T=0\}$ whereas in Lemma \ref{lemma:mix} the conditioning is on $\{Z=0,T=0\}$ or $\{Z=1,T=1\}$. In practice, this means that in step 2, the $\widehat{\Delta}_{PI}^{ER}$ (and $\widehat{\Delta}_{PI,MO}^{ER}$) estimator(s) uses more data than $\widehat{\Delta}_{PI}$ (and $\widehat{\Delta}_{PI,MO}$) to fit the mixture of experts.} In the scenario where both monotonicity and exclusion restriction hold (Scenario 4), the estimator $\widehat{\Delta}_{PI,MO}^{ER}$ achieved the lowest RMSE, close to the performance of instrumental variable estimators $\widehat{\Delta}_{IV\,matching}$ and $\widehat{\Delta}_{IV wald},$ which are specifically designed for that setting. All of our four proposed estimation methodologies outperformed instrumental variable estimators in the scenarios where either monotonicity, exclusion restriction, or both assumptions were violated. In these situations (Scenarios 2, 3, and 1 respectively), the $\widehat{\Delta}_{IV\,matching}$ and $\widehat{\Delta}_{IV wald}$ estimators exhibited high biases; while in comparison, our proposed estimators had much lower finite sample biases—despite model misspecification. For finite samples, the rate of convergence of our proposed estimators appeared close to $\sqrt{n}$ (Figure \ref{fig:sim_conv}). 
\begin{figure}
\centering
\scalebox{.55}{
\includegraphics{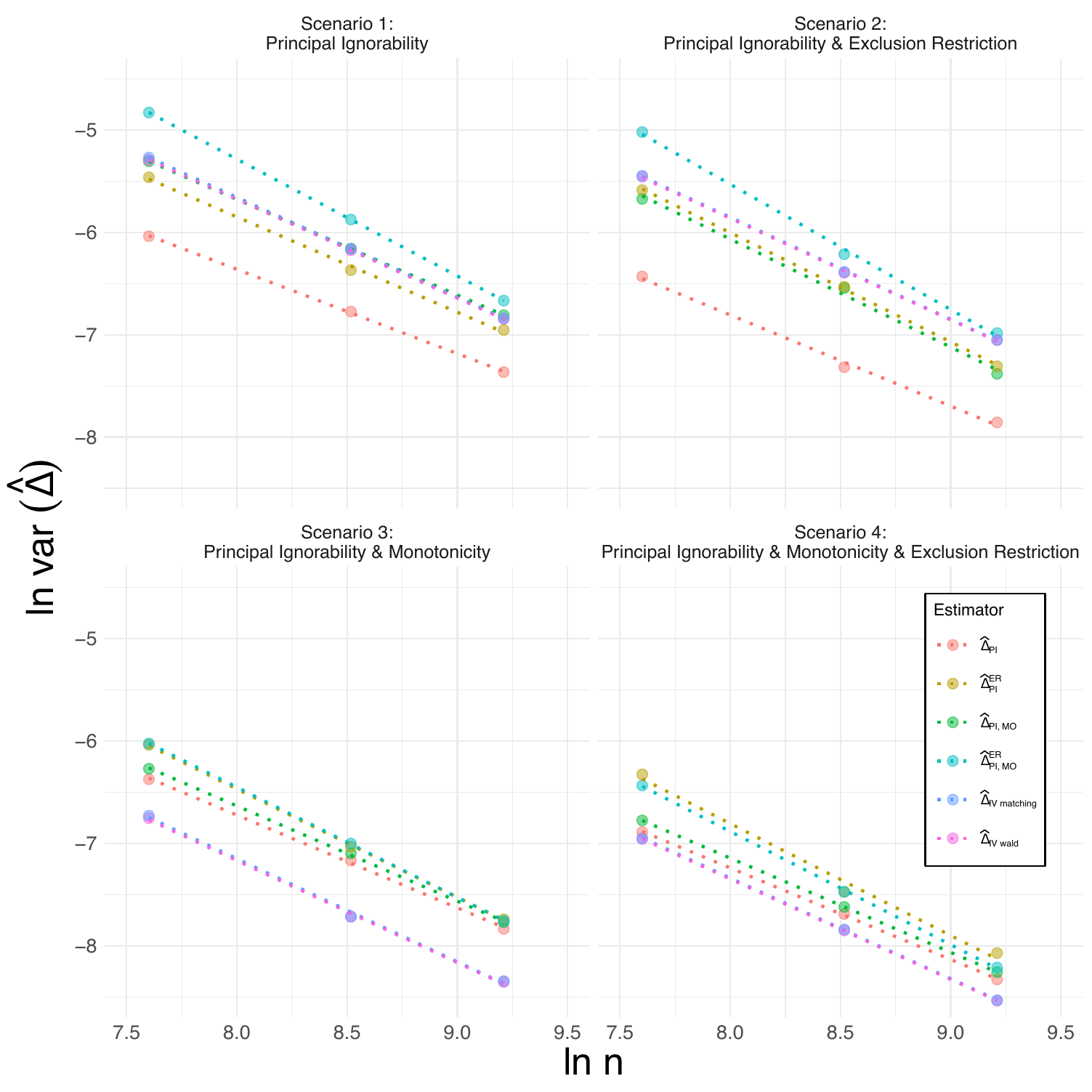}}
\caption{Estimators' variance and rate of convergence under misspecified parametric models.}
\label{fig:sim_conv}
\scriptsize For each estimator/scenario combination, slopes describe rates of convergence (e.g., a slope of $-1/2$ points to a convergence speed of $\sqrt{n}$), while intercepts approximate the logarithm of asymptotic variances.
\end{figure}
In most scenarios and sample sizes, the $\widehat{\Delta}_{PI}$ estimator exhibited the lowest variance (Table \ref{tab:sim_res} and Figure \ref{fig:sim_conv}). Our simulations with misspecified models suggest that this estimator might be a reasonable choice for CACE estimation when monotonicity or exclusion restriction assumptions cannot be confidently made. Overall, similar patterns were found for our estimation methodologies under correct model specification, though in that case, finite sample biases were lower and more substantially decreasing with sample size (Appendix \ref{app:sim_supp_res}). 

\section{Application on the Promotion of Breastfeeding Intervention Trial}\label{sec:app}
\subsection{Description}
The Promotion of Breastfeeding Intervention Trial (PROBIT) was conducted to assess the effects of a breastfeeding promotion program on infant weight at three months \citep{kramer2001}. The trial recruited mother-infant pairs from 31 Belarusian maternity hospitals and randomly assigned them to either the breastfeeding promotion program or standard care. For our experiments, we used the PROBITsim simulation learner \citep{goetghebeur2020}, which is a publicly available, anonymized database that replicates data from the original trial. In our setting, the allocated treatment corresponds to an allocation to the breastfeeding promotion program (i.e., if so, we have $Z=1$), and the treatment effectively taken describes whether a participant attended that program (i.e., in that case, we have $T=1$). Our main outcome of interest is infant weight at three months, discretized at 6000g (i.e., $Y=1$ for weights greater than 6000g). Participants' pre-randomization covariates (i.e, the variable $X$) comprise two categorical variables (location, education), four binary variables (maternal allergy, smoking status, child born by caesarian, sex of the child), and two continuous variables (mother's age at randomization, birth weight). Our objective is to estimate the CACE, which, in this context, represents the effect among those individuals who will attend the breastfeeding program when invited but not otherwise. We estimate the CACE using four estimators introduced in this paper: $\widehat{\Delta}_{PI}$, $\widehat{\Delta}_{PI}^{ER}$, $\widehat{\Delta}_{PI,MO}$, $\widehat{\Delta}_{PI,MO}^{ER}$ as well as two instrumental variable estimators: $\widehat{\Delta}_{IV wald}$ and $\widehat{\Delta}_{IV\,matching}$. We calculate 95\% confidence intervals (95\% CI) using the bootstrap with 999 replicates.

\subsection{Results}
\begin{figure}[!ht]
\centering
\includegraphics{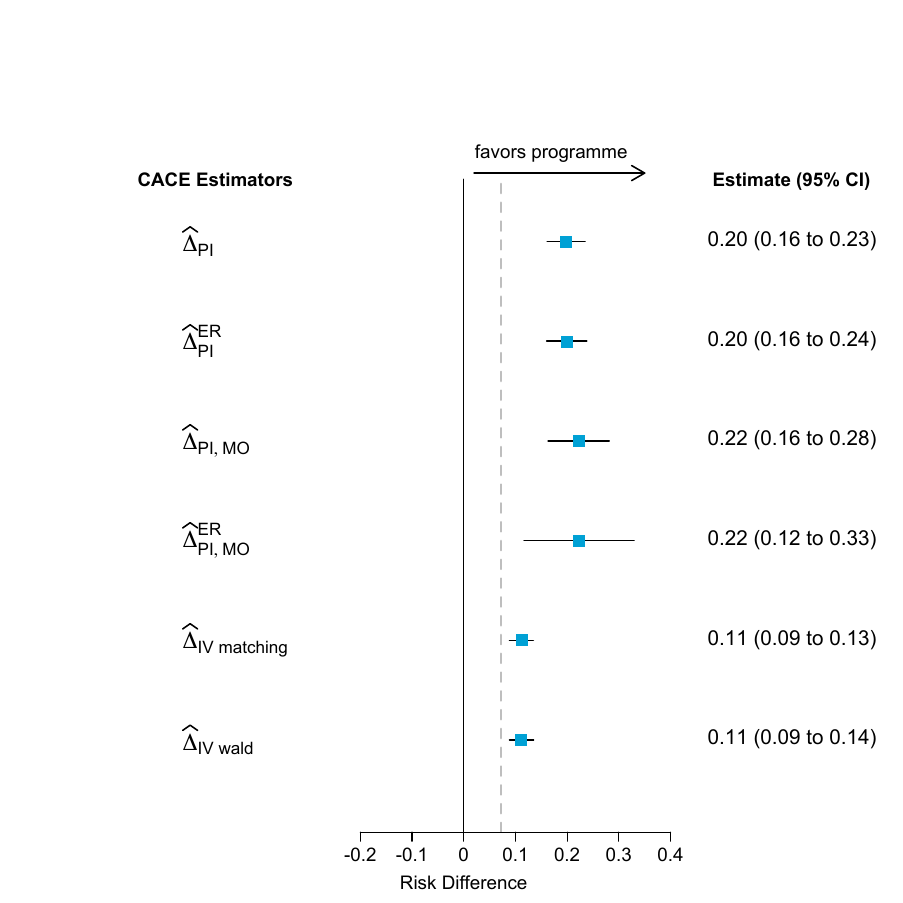}
\caption{Estimation of the CACE for the Promotion of Breastfeeding Intervention Trial.}
\label{fig:probit}
\scriptsize The gray dotted line indicates the estimated average treatment effect. Abbreviations: CACE = Complier Average Causal Effect.
\end{figure}
The estimated average treatment effect of the allocation to the breastfeeding promotion program vs the allocation to standard care was 0.07; 95\% CI [0.06 to 0.09]. Because 36\% of participants allocated to the program did not attend it, this so-called intention to treat analysis may lead to underemphasizing the intrinsic effect of the breastfeeding promotion program on infant weight at three months. In fact, the calculation of the CACE with all six estimators showed estimates greater than the estimated average treatment effect (Figure \ref{fig:probit}). The instrumental variable estimators $\widehat{\Delta}_{IV wald}$ and $\widehat{\Delta}_{IV\,matching},$ which rely on the exclusion restriction and monotonicity assumptions, both yield CACE estimates of 0.11 with tight confidence intervals (95\% CI [0.09 to 0.14] and [0.09 to 0.13] respectively). On the other hand, the estimators $\widehat{\Delta}_{PI}$, $\widehat{\Delta}_{PI}^{ER}$, $\widehat{\Delta}_{PI,MO}$, and $\widehat{\Delta}_{PI,MO}^{ER}$ leverage the principal ignorability assumption, and compared to instrumental variable estimators, yield greater CACE estimates with larger confidence intervals. The validity of the IV estimates, however, depends on assumptions that are questionable in this setting. The exclusion restriction could be violated if the encouragement from the program itself creates a placebo-like effect, leading mothers to adopt other healthy behaviors that affect infant weight, regardless of their final breastfeeding decision. Furthermore, the monotonicity assumption would fail if some mothers who would have breastfed on their own react negatively to the institutional encouragement and choose not to, acting as ``contrarians" or ``defiers." Given these plausible violations, the estimate from the $\widehat{\Delta}_{PI}$ estimator (0.20; 95\% CI [0.16 to 0.23]), which does not require these assumptions, may provide a more reliable evaluation of the intrinsic effect of the program. Furthermore, because all six estimators point to clinically meaningful and statistically significant effects, we can reasonably conclude that, in compliers, the breastfeeding promotion program causes greater infant weight at three months.

\section{Discussion}\label{sec:diss}
We have introduced a causal inference framework based on a mixture of experts architecture to estimate the CACE. Our approach provides a flexible alternative to classical instrumental variable methods by replacing the behavioral assumptions of monotonicity and exclusion restriction with a parametric model grounded in principal ignorability. The framework yields four distinct estimators, applicable to both experimental and observational data, and is particularly useful in settings where traditional IV assumptions are likely to be violated.

A central element of our methodology is the trade-off it presents. The principal ignorability assumption, while allowing us to achieve point identification in a more general setting, is a strong condition. The identifiability of our model's components, which we prove under parametric assumptions, is crucial for ensuring the learning procedure is well-posed. We also note that our two-step estimation process, which separates the training of the gating and expert networks, is essential for identifiability; a joint estimation would fail.

This work opens several avenues for future research. A natural next step is to relax the parametric assumptions of our model. While we provide initial non-parametric implementations in Appendix \ref{app:algos}, establishing formal nonparametric point identification would require strong regularity conditions on the expert and gating networks. Deriving a set of sufficient regularity conditions is mathematically challenging and represents an important direction for theoretical work. A more immediately practical direction is to integrate our approach with the double/debiased machine learning framework \citep{chernozhukov2018}. Using cross-fitting with flexible machine learning algorithms to estimate the nuisance functions could yield semi-parametrically efficient CACE estimators that are robust to own-observation bias. Finally, the mixture of experts structure is highly adaptable and could be extended to handle other complexities, such as censored outcomes through mixtures of survival models \citep{kuo2000}.

\acks{We are grateful for enlightening conversations with Antoine Chambaz, Julie Josse, Bénédicte Colnet, and Alex Fernandes.

FG conceived the study. FG and CB wrote the codes and did the simulation and application analyses. FG and FP worked on the mathematical proofs. FG drafted the manuscript with inputs from CB, FP and RP. All the authors read the paper and suggested edits. RP and FP supervised the project. FG and CB accessed and verified the data. 

RP acknowledges the support of the French Agence Nationale de la Recherche as part of the “Investissements d'avenir” program, reference ANR-19- P3IA-0001 (PRAIRIE 3IA Institute). FP acknowledges support from the French Agence Nationale de la Recherche through the project reference ANR-22-CPJ1-0047-01.

The authors have disclosed that they do not have any conflicts of interest related to this article.}

\clearpage
\appendix
\section{}\label{app:notations}

\begin{table}[h!]
\caption{Summary of notations}
\centering
\label{tab:notations}
\scalebox{0.807}{
\begin{tabular}{lp{0.75\linewidth}}
\hline\hline
\textbf{Notation}                & \textbf{Meaning}                                       \\ \hline
$X$                     & Baseline (i.e., pre-randomization) covariates \\
$Z$                     & Allocated binary treatment                    \\
$S_c$                   & Binary indicator whether an individual is a complier\\
$S_a$                   & Binary indicator whether an individual is an always-taker\\
$S_n$                   & Binary indicator whether an individual is a never-taker\\
$S_d$                   & Binary indicator whether an individual is a defier\\
$\boldsymbol{S}=(S_c,S_a,S_n,S_d)^T$ & One-hot-encoded latent stratum of an individual    \\
$\{T^{s=k}\}_{k\in\{c,a,n,d\}}$ & Potential treatment under the stratum $k$\\
$T$                     & Treatment effectively taken                   \\
$\{Y^{s=k,z=l,t=m}\}$ & Elementary potential outcome                  \\
$\{Y^{t=m}\}_{m\in\{0,1\}}$ & Potential outcome                             \\
$Y$                     & Observed outcome                              \\ 
$\Delta=\EX(Y^{t=1}-Y^{t=0}|S_c=1)$ & Complier Average Causal Effect (CACE)    \\
$\eta(x)=\EX(Z|X=x)$      & Allocation ratio function \\
$e(x)=\EX(T|X=x)$         & Standard propensity score\\
$\pi(x,z)=\EX(T|X=x,Z=z)$   & Propensity score adaptation in the context of imperfect compliance\\
$\rho_k(x)= \EX(S_k|X=x)$ &  Probability functions that an individual with covariates $x$ is in the stratum $k$\\
$\mu_k(z)=\EX(T^{s=k}|Z=z)$ & Known expert from a mixture of expert model\\
$P_{klm}(x)=\EX(S_k|Z=l, T=m, X=x)$ & Probability functions that an individual with allocated treatment $l$, treatment taken $m$, and covariates $x$ is in the stratum $k$\\
$q_{lm}(x)=\EX(Y|Z=l,T=m,X=x)$ & Conditional observed outcome functions\\
$Q_{klm}(x)=\EX(Y^{s=k,z=l,t=m}|X=x)$ & Conditional elementary potential outcome functions\\
$P_{k\cdot m}(x)= \EX(S_k|T=m,X=x)$ & Probability functions that an individual with treatment taken $m$, and covariates $x$ is in the stratum $k$\\
$q_{\cdot m}(x)= \EX(Y|T=m,X=x)$ & Conditional observed outcome functions, irrespective of $Z$\\
$Q_{k}(x)= \EX(Y^{s=k}|X=x)$ & Conditional potential outcome functions under exclusion restriction\\
$\mathcal{X}=\R^d$      & Space of covariates\\
$\mathcal{Z}=\{0,1\}$   & Space of allocated treatments\\
$\mathcal{P}$           & Probability distribution of the experiment\\
$k\in\{c,a,n,d\}$   & Index for a particular stratum\\
$l\in\{0,1\}$       & Index for a particular allocated treatment\\
$m\in\{0,1\}$       & Index for a particular treatment taken\\
$\widehat{\Delta}_{PI}$ & CACE estimator assuming principal ignorability\\
$\widehat{\Delta}_{PI}^{ER}$ & CACE estimator assuming principal ignorability and exclusion restriction\\
$\widehat{\Delta}_{PI,MO}$ & CACE estimator assuming principal ignorability and monotonicity\\
$\widehat{\Delta}_{PI,MO}^{ER}$ & CACE estimator assuming principal ignorability, exclusion restriction, and monotonicity\\
$\widehat{\Delta}_{IV\,matching}$ & CACE instrumental variable estimator using matching\\
$\widehat{\Delta}_{IV\,wald}$ & Standard CACE instrumental variable estimator\\
\hline\hline
\end{tabular}
}
\end{table}

\section{Proofs}\label{app:proofs}
\subsection{Proof of Lemma \ref{lemma:est}}
\begin{proof}
\begin{align*}
\Delta&\eqdef \EX[Y^{t=1}-Y^{t=0}|S_c=1] \\
&=\EX[Y^{s=c,z=1,t=1}-Y^{s=c,z=0,t=0}|S_c=1] \\
&=\EX \! \Big[\EX\big(Y^{s=c,z=1,t=1}-Y^{s=c,z=0,t=0}|X,S_c\big)|S_c=1\Big] \\
&\mathrel{\hspace{-0.38cm} \stackrel{\mathrm{Asm.} \ref{as:pi}\emph{(ii)}}{=}}\EX \! \Big[\EX\big(Y^{s=c,z=1,t=1}-Y^{s=c,z=0,t=0}|X\big)|S_c=1\Big]\\ \
&=\EX \! \Big[ \!\EX(Y^{s=c,z=1,t=1}|X)\!- \!\EX(Y^{s=c,z=0,t=0}|X)|S_c=1\Big] \\
&=\EX\Big[ Q_{c11}(X)-Q_{c00}(X)|S_c=1\Big] \\
&=\int_{\mathcal{X}}\big\{ Q_{c11}(x)-Q_{c00}(x)\big\}p_{X|S_c}(x|1)dx \\
&=\int_{\mathcal{X}}\big\{ Q_{c11}(x)-Q_{c00}(x)\big\}\frac{p_{S_c|X}(1|x)}{p_{S_c}(1)}p_X(x)dx \\
&=\EX\Bigg[\big\{ Q_{c11}(X)-Q_{c00}(X)\big\}\frac{\EX(S_c|X)}{\EX(S_c)}\Bigg]  \\
&=\frac{\EX\Big[\big\{ Q_{c11}(X)-Q_{c00}(X)\big\}\EX(S_c|X)\Big]}{\EX(S_c)} \\
&=\frac{\EX\Big[\big\{Q_{c11}(X)-Q_{c00}(X)\big\}\EX(S_c|X)\Big]}{ \EX[\EX(S_c|X)]} \\
&=\frac{\EX\Big[\big\{ Q_{c11}(X)-Q_{c00}(X)\big\}\rho_c(X)\Big]}{ \EX[\rho_c(X)]}
\end{align*}
Assumption \ref{as:pos} guarantees that the conditioning on $\{S_c=1\}$ is well-defined, and that $\EX[p_c(X)]\neq0$.
\end{proof}

\subsection{Proof of Lemma \ref{lemma:mix4}}
\begin{proof}
\begin{align*}
\pi(X,Z)&\eqdef \EX(T|X,Z) \notag \\
&\hspace{-0.25cm}\stackrel{ \text{Def. \ref{def_ttt}}}{=}\EX(S_cT^{s=c} + S_aT^{s=a}  + S_nT^{s=n}+ S_dT^{s=d}|X,Z)\\
&=\EX(S_cZ + S_a\times1 + S_n\times0 + S_d(1-Z)|X,Z) \notag\\
&=Z\EX(S_c|X,Z) + 1\times \EX(S_a|X,Z) + 0\times \EX(S_n|X,Z)+ (1-Z)\EX(S_d|X,Z) \notag\\
&\hspace{-0.38cm} \stackrel{\text{Asm. \ref{as:pi}\emph{(i)}}}{=}Z\EX(S_c|X) + 1\times \EX(S_a|X) + 0\times \EX(S_n|X)+ (1-Z)\EX(S_d|X) \\
&=\EX(Z|Z)\EX(S_c|X) + \EX(1|Z) \EX(S_a|X) + \EX(0|Z) \EX(S_n|X)+ \EX(1-Z|Z)\EX(S_d|X) \notag\\
&=\EX(S_c|X)\EX(T^{s=c}|Z) + \EX(S_a|X)\EX(T^{s=a}|Z)  + \EX(S_n|X)\EX(T^{s=n}|Z)+ \EX(S_d|X)\EX(T^{s=d}|Z) \notag\\
&=\sum_{k\in\{c,a,n,d\}}\rho_k(X)\mu_k(Z) \notag
\end{align*}
Assumption \ref{as:exo} guarantees that the conditioning on $\{X,Z\}$ is well-defined.
\end{proof}

\subsection{Proof of Lemma \ref{lemma:mix}}
\begin{proof}
\begin{align}
q_{11}(X) &\eqdef \EX(Y|Z=1, T=1, X) \notag \\
    &\hspace{-0.26cm}\stackrel{\text{Asm. \ref{as:consistency}}}{=}\EX(TY^{t=1}+(1-T)Y^{t=0}|Z=1, T=1, X)\notag\\
    &=\EX(Y^{t=1}|Z=1, T=1, X) \notag \\
    &\hspace{-0.2cm}\stackrel{\text{Def. \ref{def_po}}}{=} \EX(S_cY^{s=c,z=1,t=1}|Z=1, T=1, X) + \EX(S_aY^{s=a,z=1,t=1}|Z=1, T=1, X) \notag \\
    &=\EX(S_c|Z=1, T=1, X) \times \EX(Y^{s=c,z=1,t=1}|Z=1, T=1, X) \notag\\ 
    &\quad +\EX(S_a|Z=1, T=1, X) \times \EX(Y^{s=a,z=1,t=1}|Z=1, T=1, X) \label{eq:proof:mix}\\
    &=\EX(S_c|Z=1, T=1, X)\EX(Y^{s=c,z=1,t=1}|X) \notag\\
    &\quad +\EX(S_a|Z=1, T=1, X)\EX(Y^{s=a,z=1,t=1}|X) \label{eq:proof_bis:mix}\\
    &=P_{c11}(X)Q_{c11}(X) + P_{a11}(X)Q_{a11}(X) \notag
\end{align}
Assumptions \ref{as:exo}, \ref{as:pi}\emph{(i)}, and \ref{as:pos} guarantees that the conditioning on $\{Z=1,T=1,X\}$ is well-defined. The equality in \eqref{eq:proof:mix} rely on Assumption \ref{as:pi}\emph{(ii)}. The equalities in \eqref{eq:proof:mix} and \eqref{eq:proof_bis:mix} rely on the fact that the extra conditioning on $\{Z=1,T=1\}$ does not open a backdoor path (this can be verified by appealing to a d-separation argument on the graph given in Appendix \ref{app:graph}). The proof for the $q_{00}(X)$ formula follows a similar argument.
\end{proof}

\subsection{Proof of Lemma \ref{lemma:P}}
\begin{proof}
\begin{align*}
P_{c11}(X)&\eqdef \EX(S_c|Z=1, T=1, X)\\
&=\Prob(S_c=1|Z=1, T=1, X)\\
&=\Prob(T=1|S_c=1, Z=1, X) \times \frac{\Prob(S_c=1| Z=1, X)}{\Prob(T=1| Z=1, X)}\\
&\mathrel{\hspace{-0.181cm}\stackrel{\mathrm{Def.} \ref{def_ttt}}{=}}\EX(S_cZ \! + \! S_a \! + \! S_d(1-Z)|S_c=1, Z=1, X) \times \frac{\Prob(S_c=1| Z=1, X)}{\Prob(T=1| Z=1, X)}\notag\\
&=\frac{\Prob(S_c=1| Z=1, X)}{\Prob(T=1| Z=1, X)}\\
& \mathrel{\hspace{-0.38cm}\stackrel{\mathrm{Asm. \ref{as:pi}}\emph{(i)}}{=}}\frac{\Prob(S_c=1| X)}{\EX(T| Z=1, X)} \notag\\
&\mathrel{\hspace{-0.181cm}\stackrel{\mathrm{Def.} \ref{def_ttt}}{=}}\frac{\EX(S_c| X)}{\EX(S_c| Z=1, X)+\EX(S_a| Z=1, X)} \notag\\
&\mathrel{\hspace{-0.38cm}\stackrel{\mathrm{Asm. \ref{as:pi}}\emph{(i)}}{=}}\frac{\EX(S_c| X)}{\EX(S_c| X)+\EX(S_a| X)} \notag\\
&=\rho_c(X)/\{\rho_c(X)+\rho_a(X)\} .
\end{align*}
Assumptions \ref{as:exo}, \ref{as:pi}\emph{(i)}, and \ref{as:pos} guarantees that the conditioning on $\{Z=1,T=1,X\}$ is well-defined. The proof for the $P_{a11}$, $P_{c00}$, and $P_{n00}$ formulas follows a similar argument.
\end{proof}

\subsection{Proof of Lemma \ref{lemma:mix_er}}
\begin{proof}
\begin{align}
q_{\cdot1}(X)&\eqdef \EX(Y| T=1, X) \notag \\
    &\mathrel{\hspace{-0.29cm}\stackrel{\mathrm{Asm.} \, \ref{as:consistency}}{=}}\EX(TY^{t=1}+(1-T)Y^{t=0}|T=1, X) \notag \\
    &=\EX(Y^{t=1}|T=1, X) \notag \\
    &\mathrel{\hspace{-0.29cm}\stackrel{\mathrm{Asm.} \, \ref{as:er}}{=}}\EX(S_cY^{s=c,z=1,t=1} + S_aY^{s=a} + S_dY^{s=c,z=0,t=1}|T=1, X) \notag \\
    &=\EX(S_c|T=1, X)\EX(Y^{s=c,z=1,t=1}|T=1, X) \label{eq:proof:eq_mix} \\ 
    &\quad +\EX(S_a| T=1, X)\EX(Y^{s=a}|T=1, X) \notag \\ 
    &\quad +\EX(S_d| T=1, X)\EX(Y^{s=c,z=0,t=1}|T=1, X) \notag \\
    &=\EX(S_c|T=1, X)\EX(Y^{s=c,z=1,t=1}|X) \label{eq:proof_bis:eq_mix}\\ 
    &\quad+\EX(S_a| T=1, X)\EX(Y^{s=a}|X) \notag\\ 
    &\quad +\EX(S_d| T=1, X)\EX(Y^{s=c,z=0,t=1}|X) \notag\\
    &=P_{c\cdot 1}(X)Q_{c11}(X) + P_{a\cdot 1}(X)Q_{a}(X)+ P_{d\cdot1}(X)Q_{d01}(X) \notag
\end{align}
Assumptions \ref{as:exo} and \ref{as:pos} guarantees that the conditioning on $\{T=1,X\}$ is well-defined. The equality in \eqref{eq:proof:eq_mix} rely on Assumption \ref{as:pi} \emph{(ii)}. The equalities in \eqref{eq:proof:eq_mix} and \eqref{eq:proof_bis:eq_mix} rely on the fact that the extra conditioning on $T=1$ does not open a backdoor path (this can be verified by appealing to a d-separation argument on the graph given in subsection \ref{app:graph}). The proof for the $q_{\cdot0}(X)$ formula follows a similar argument.
\end{proof}

\subsection{Proof of Lemma \ref{lemma:P_er}}
\begin{proof}
\begin{align}
P_{c\cdot 1}(X) &\eqdef \EX(S_c|T=1,X) \notag \\
&=\Prob(Z=0|T= \! 1,X)\overbrace{\Prob(S_c=1|Z=0,T= \!1,X)}^{=0}  \notag \\
&+ \Prob(Z=1|T= \! 1,X)\Prob(S_c=1|Z=1,T= \!1,X)\notag \\
&=\Prob(Z=1|T=1,X)P_{c11}(X)\notag \\
&=\frac{\Prob(Z=1|X)}{\Prob(T=1|X)}\Prob(T=1|Z=1,X)P_{c11}(X) \notag\\
&\mathrel{ \hspace{-0.17cm}\stackrel{\mathrm{def}. \ref{def_ttt}}{=}} \, \frac{\eta(X)}{e(X)}\EX(S_cZ+S_a+S_d(1-Z)|Z=1,X) P_{c11}(X) \notag\\
&=\frac{\eta(X)}{e(X)}\big\{\EX(S_c|Z=1,X)+\EX(S_a|Z=1,X)\big\}  P_{c11}(X) \notag\\
&\mathrel{\hspace{-0.44cm} \stackrel{\mathrm{Asm.} \, \ref{as:pi} \, \emph{(i)}}{=}} \,\frac{\eta(X)}{e(X)}\big\{\EX(S_c|X)+\EX(S_a|X)\big\}P_{c11}(X) \notag\\
&\mathrel{\hspace{-0.28cm}\stackrel{\mathrm{Lem.} \, \ref{lemma:P}}{=}} \,\frac{\eta(X)}{e(X)}\big\{\rho_c(X)+\rho_a(X)\big\}
\frac{\rho_c(X)}{\rho_c(X)+\rho_a(X)}\notag\\
&=\frac{\eta(X)}{e(X)}\rho_c(X) \notag
\end{align}
Assumptions \ref{as:exo} and \ref{as:pos} guarantees that the conditioning on $\{T=1,X\}$ is well-defined. The proofs for the $P_{a\cdot 1}(X)$, $P_{d\cdot 1}(X)$, $P_{c\cdot 0}(X)$, $P_{n\cdot 0}(X)$, and $P_{d\cdot 0}(X)$ formulas follow similar arguments. The standard propensity score can be further represented as
\begin{align}
    e(X)\eqdef&\EX(T|X) \notag\\
    =&\EX(S_cZ+S_a+S_d(1-Z)|X) \notag\\
    \stackrel{\mathrm{Asm.} \, \ref{as:pi} \, \emph{(i)}}{=}&\EX(S_c|X)\EX(Z|X) + \EX(S_a|X) + \EX(S_d|X)\EX(1-Z|X) \notag\\
    =& \rho_c(X)\eta(X) + \rho_a(X) + \rho_d(X)\{1-\eta(X)\}. \notag
\end{align}
\end{proof}

\subsection{Probabilistic graphical model}\label{app:graph}
Below, we provide the probabilistic graphical model corresponding to the data generating mechanism described in section \ref{sec:setup}. Conditional independences between variables can be read from this diagram using the rules of d-separation. For clarity we use random vector notation 
\begin{align*}
\boldsymbol{Y^{s,z,t}}=(&Y^{s=c,z=1,t=1},Y^{s=a,z=1,t=1},Y^{s=a,z=0,t=1},
Y^{s=d,z=0,t=1}, Y^{s=c,z=0,t=0},Y^{s=n,z=1,t=0},\\
&Y^{s=n,z=0,t=0},Y^{s=d,z=1,t=0})^T    
\end{align*}
and  $\boldsymbol{Y^t}=(Y^{t=1},Y^{t=0})^T.$
\begin{center}
\begin{tikzpicture}
  \node (X) at (0, 8) {X};
  \node (S) at (0, 6) {$\boldsymbol{S}$};   \node (Y^{s,z,t}) at (2, 6) {$\boldsymbol{Y^{s,z,t}}$};
  \node (Z) at (0, 4) {Z};
  \node (T) at (0, 2) {T};   \node (Y^t) at (2, 2) {$\boldsymbol{Y^t}$};
  \node (Y) at (0, 0) {Y};

  \draw[->] (X) -- (S);      
  \draw[->] (X) -- (Y^{s,z,t});
  \draw[->, bend right] (X) to node[midway, left] {} (Z);
  \draw[->, bend right] (S) to node[midway, left] {} (T);
  \draw[->] (Y^{s,z,t}) -- (Y^t);
  \draw[->] (S) -- (Y^t);
  \draw[->] (Z) -- (Y^t);
  \draw[->] (Z) -- (T);
  \draw[->] (Y^t) -- (Y);
  \draw[->] (T) -- (Y);
\end{tikzpicture}
\end{center}

\section{EM algorithms}\label{app:algos}
All EM algorithms in this section are adaptations of the algorithm provided in \citet{jordan1994}, based on the description from \citet{Xu1993}. Algorithms \ref{algo:2}, \ref{algo:3}, \ref{algo:5}, \ref{algo:6}, \ref{algo:7}, \ref{algo:8}, and \ref{algo:10} take a data subset as input. To avoid clutter, we drop subscripts of the form $(-)_{i:Z_i=l,T_i=l}$ in these algorithms. However, as a reminder of data subsetting, we note sums over $n^\prime$ elements rather than $n$.

\begin{algorithm}[H]
\caption{The EM procedure for estimating $\rho_s(\cdot;\delta_s)$ where $s \in \{c,a,n,d \}$. }\label{algo:1}
\begin{algorithmic}
\begin{footnotesize}
\Input Data $(X_i^\rho, Z_i, T_i)_{1 \leq i \leq n}$ where $X_i^{\rho}$ is a relevant subset of the variables contained in $X_i$. 

\Initialize the prior probabilities associated with the nodes of the tree as
\begin{equation*}
g_{c,i} \gets 1/4, \quad g_{a,i} \gets 1/4, \quad g_{n,i} \gets 1/4, \quad \text{and} \quad g_{d,i} \gets 1/4.    
\end{equation*}

\Compute individual contributions to each expert's likelihood as
\begin{align*}
    &L_{c,i} \gets Z_iT_i + (1-Z_i)(1-T_i)\\
    &L_{a,i} \gets T_i\\
    &L_{n,i} \gets 1-T_i\\
    &L_{d,i} \gets Z_i(1-T_i) + (1-Z_i)T_i
\end{align*}

\Iterate until convergence on the parameters $\delta=(\delta_c^T,\delta_a^T,\delta_n^T,\delta_d^T)^T$: \\
Compute the posterior probabilities associated with the nodes of the tree as \Comment{E-step}
\begin{align*}
    &h_{c,i} \gets g_{c,i}L_{c,i} / \sum_{s\in\{c,a,n,d\}} g_{s,i}L_{s,i} \\
    &h_{a,i} \gets g_{a,i}L_{a,i} / \sum_{s\in\{c,a,n,d\}} g_{s,i}L_{s,i} \\
    &h_{n,i} \gets g_{n,i}L_{n,i} / \sum_{s\in\{c,a,n,d\}} g_{s,i}L_{s,i} \\
    &h_{d,i} \gets g_{d,i}L_{d,i} / \sum_{s\in\{c,a,n,d\}} g_{s,i}L_{s,i} 
\end{align*}
For the gating network $\rho(\cdot)$ estimate parameters $\delta$ by solving the IRLS problem\Comment{M-step}\\ 
\begin{equation*}
    \delta \gets \argmax\limits_{\delta} \sum_{i=1}^n \sum_{s\in\{c,a,n,d\}} h_{s,i} \ln{ \Big (\frac{\exp{\delta_s^TX_i^\rho} }{ \sum_{k\in\{c,a,n,d\}}\exp{\delta_k^TX_i^\rho}} \Big)} 
\end{equation*}
$\quad$ as a multinomial logistic regression with features $(X_i^{\rho})_{1 \leq i \leq n}$,\\
$\quad$ and targets $(h_{c,i}, h_{a,i}, h_{n,i}, h_{d,i})_{1 \leq i \leq n}.$ \\
Update the prior probabilities associated with the nodes of the tree as 
\begin{align*}
    g_{c,i} \gets \exp{\delta_c^TX_i^\rho} / \sum_{k\in\{c,a,n,d\}}\exp{\delta_k^TX_i^\rho} \\
    g_{a,i} \gets \exp{\delta_a^TX_i^\rho} / \sum_{k\in\{c,a,n,d\}}\exp{\delta_k^TX_i^\rho} \\
    g_{n,i} \gets \exp{\delta_n^TX_i^\rho} / \sum_{k\in\{c,a,n,d\}}\exp{\delta_k^TX_i^\rho} \\
    g_{d,i} \gets \exp{\delta_d^TX_i^\rho} / \sum_{k\in\{c,a,n,d\}}\exp{\delta_k^TX_i^\rho}
\end{align*}

\Endloop $\rho_s(x;\hat{\delta})=\exp{\delta_s^Tx} / \sum_{k\in\{c,a,n,d\}}\exp{\delta_k^Tx}.$
\end{footnotesize}
\end{algorithmic}
\end{algorithm}

\begin{algorithm}[H]
\caption{The EM procedure for estimating $Q_c(\cdot;\zeta)$ when $Y$ is binary. Here, $Q_c$ denotes either $Q_{c11}$ or $Q_{c00}$ depending on the data subset considered.}\label{algo:2}
\begin{algorithmic}
\begin{footnotesize}
\Input Data $(X_i^{\zeta}, Y_i,P_{cll}(X_i))_{i:Z_i=l,T_i=l}$ where $l\in\{0,1\}$ and $X_i^{\zeta}$ is a relevant subset of the variables contained in $X_i.$

\Initialize the prior probabilities associated with the nodes of the tree as
\begin{equation*}
g_{c,i} \gets P_{cll}(X_i), \quad
g_{nc,i} \gets 1 - P_{cll}(X_i).    
\end{equation*}

\Initialize the parameters $(\zeta_c,\zeta_{nc})$ of the experts $\big(Q_c(\cdot),Q_{nc}(\cdot)\big)$ at random e.g., 
\begin{equation*}
\zeta_c \sim \mathcal{N}(0, D) \quad
\zeta_{nc} \sim \mathcal{N}(0, D)
\end{equation*}
with $D$ a diagonal matrix.

\Compute individual predictions from the initiated expert networks $\big(Q_c(\cdot),Q_{nc}(\cdot)\big)$ as 
\begin{equation*}
Q_{c,i} \gets \expit(\zeta_c^TX_i^\zeta) \quad
Q_{nc,i} \gets \expit(\zeta_{nc}^TX_i^\zeta)
\end{equation*}

\Iterate until convergence on the parameters $\zeta=(\zeta_c^T,\zeta_{nc}^T)^T$: \\
Compute individual contributions to each expert's likelihood as
\begin{align*}
    &L_{c,i} \gets Q_{c,i}^{Y_i} (1-Q_{c,i})^{1-Y_i}\\
    &L_{nc,i} \gets Q_{nc,i}^{Y_i} (1-Q_{nc,i})^{1-Y_i}
\end{align*}
Compute the posterior probabilities associated with the nodes of the tree as \Comment{E-step}
\begin{align*}
    &h_{c,i} \gets g_{c,i}L_{c,i} / \Big(g_{c,i}L_{c,i} + g_{nc,i}L_{nc,i} \Big) \\
    &h_{nc,i} \gets 1 - h_{nc,i}
\end{align*}
For the expert network $Q_c(\cdot)$ estimate parameters $\zeta_c$ by solving the IRLS problem \Comment{M-step}\\   
\begin{align*}
    \zeta_c \gets &\argmax\limits_{\zeta_c} \sum_{i=1}^{n'} h_{c,i} \Big[ Y_i\ln{\{ \expit(\zeta_c^TX_i^\zeta) \}} + (1-Y_i)\ln{ \{1-\expit(\zeta_c^TX_i^\zeta) \} }\Big] 
\end{align*}
$\quad$ as a weighted logistic regression with features $X_i^{\zeta}$, targets $Y_i$ and weights $h_{c,i}$. \\
For the expert network $Q_{nc}(\cdot)$ estimate parameters $\zeta_{nc}$ by solving the IRLS problem
\begin{align*}
    \zeta_{nc} \gets &\argmax\limits_{\zeta_{nc}} \sum_{i=1}^{n'} h_{nc,i} \Big[ Y_i\ln{\{ \expit(\zeta_{nc}^TX_i^\zeta) \}}  + (1-Y_i)\ln{ \{1-\expit(\zeta_{nc}^TX_i^\zeta) \} }\Big] 
\end{align*}
$\quad$ as a weighted logistic regression with features $X_i^{\zeta}$, targets $Y_i$ and weights $h_{nc,i}$. \\
Update the predictions from the expert networks $\big(Q_c(\cdot),Q_{nc}(\cdot)\big)$ as 
\begin{equation*}
Q_{c,i} \gets \expit(\zeta_c^TX_i^\zeta) \quad
Q_{nc,i} \gets \expit(\zeta_{nc}^TX_i^\zeta)
\end{equation*}
\Endloop $Q_c(x;\hat{\zeta})=\expit{(\zeta_c^Tx)}$
\end{footnotesize}
\end{algorithmic}
\end{algorithm}

\begin{algorithm}[H]
\caption{The EM procedure for estimating $Q_c(\cdot;\zeta)$ when $Y$ is continuous. Here, $Q_c$ denotes either $Q_{c11}$ or $Q_{c00}$ depending on the data subset considered.}\label{algo:3}
\begin{algorithmic}
\begin{footnotesize}
\Input Data $(X_i^{\zeta }, Y_i,P_{cll}(X_i))_{i:Z_i=l,T_i=l}$ where $l\in\{0,1\}$ and $d$-dimensional $X_i^{\zeta}$ is a relevant subset of the variables contained in $X_i.$

\Initialize the prior probabilities associated with the nodes of the tree as
\begin{equation*}
g_{c,i} \gets P_{cll}(X_i), \quad
g_{nc,i} \gets 1 - P_{cll}(X_i).     
\end{equation*}

\Initialize the parameters $\Big((\zeta_c,\sigma_c^2),(\zeta_{nc},\sigma_{nc}^2)\Big)$ of the experts $\big(Q_c(\cdot),Q_{nc}(\cdot)\big)$ at random e.g., 
\begin{align*}
\zeta_c \sim \mathcal{N}(0, D) \quad  \sigma_c^2 \gets 1 \\
\zeta_{nc} \sim \mathcal{N}(0, D)
\quad \sigma_{nc}^2 \gets 1 
\end{align*}
with $D$ a diagonal matrix.

\Compute individual predictions from the initiated expert networks $\big(Q_c(\cdot),Q_{nc}(\cdot)\big)$ as 
\begin{equation*}
Q_{c,i} \gets \zeta_c^TX_i^\zeta \quad
Q_{nc,i} \gets \zeta_{nc}^TX_i^\zeta
\end{equation*}
\Iterate until convergence on the parameters $\zeta=(\zeta_c^T,\zeta_{nc}^T)^T$: \\
Compute individual contributions to each expert's likelihood as
\begin{align*}
    L_{c,i} &\gets \mathcal{N}_\mathcal{L}(Y_i|\mu=Q_{c,i},\sigma^2=\sigma_c^2)\\
    L_{nc,i} &\gets \mathcal{N}_\mathcal{L}(Y_i|\mu=Q_{nc,i},\sigma^2=\sigma_{nc}^2)
\end{align*}
Compute the posterior probabilities associated with the nodes of the tree as \Comment{E-step}
\begin{align*}
    &h_{c,i} \gets g_{c,i}L_{c,i} / \Big(g_{c,i}L_{c,i} + g_{nc,i}L_{nc,i} \Big)
    &h_{nc,i} \gets 1 - h_{nc,i}
\end{align*}
For the expert network $Q_c(\cdot)$ estimate parameters $\zeta_c$ by solving the\\ 
$\quad$ weighted least-squares problem \Comment{M-step}   
\begin{equation*}
    \zeta_c \gets \argmin\limits_{\zeta_c} \sum_{i=1}^{n'} h_{c,i} \big( Y_i - \zeta_c^TX_i^{\zeta}\big)^2 
\end{equation*}
$\quad$ as a weighted linear regression with features $X_i^{\zeta}$, targets $Y_i$ and weights $h_{c,i}$. \\
For the expert network $Q_{nc}(\cdot)$ estimate parameters $\zeta_{nc}$ by solving the\\
$\quad$ weighted least-squares problem \\   
\begin{equation*}
    \zeta_{nc} \gets \argmin\limits_{\zeta_{nc}} \sum_{i=1}^{n'} h_{nc,i} \big( Y_i - \zeta_{nc}^TX_i^{\zeta}\big)^2 
\end{equation*}
$\quad$ as a weighted linear regression with features $X_i^{\zeta}$, targets $Y_i$ and weights $h_{nc,i}$. \\
Update the variance parameter for each expert $\big(Q_c(\cdot),Q_{nc}(\cdot)\big)$ as
\begin{align*}
\sigma_c^2 &\gets  \frac{1}{n'-d}\sum_{i=1}^{n'} h_{c,i}\big( Y_i - \zeta_c^TX_i^{\zeta}\big)^2, \\
\sigma_{nc}^2 &\gets  \frac{1}{n'-d}\sum_{i=1}^{n'} h_{nc,i}\big( Y_i - \zeta_{nc}^TX_i^{\zeta}\big)^2
\end{align*}
Update the predictions from the expert networks $\big(Q_c(\cdot),Q_{nc}(\cdot)\big)$ as 
\begin{equation*}
Q_{c,i} \gets \zeta_c^TX_i^\zeta \quad
Q_{nc,i} \gets \zeta_{nc}^TX_i^\zeta
\end{equation*}
\Endloop $Q_c(x;\hat{\zeta})=\zeta_c^Tx$
\end{footnotesize}
\end{algorithmic}
\end{algorithm}

\begin{algorithm}[H]
\caption{The nonparametric EM-like procedure for estimating $\rho_s(\cdot)$ where $s \in \{c,a,n,d \}$. }\label{algo:4}
\begin{algorithmic}
\begin{footnotesize}
\Input Data $(X_i^{\rho}, Z_i, T_i)_{1 \leq i \leq n}$ where $X_i^{\rho}$ is a relevant subset of the variables contained in $X_i$. 

\Initialize the prior probabilities associated with the nodes of the tree as
\begin{equation*}
g_{c,i} \gets 1/4, \quad g_{a,i} \gets 1/4, \quad g_{n,i} \gets 1/4, \quad \text{and} \quad g_{d,i} \gets 1/4.    
\end{equation*}

\Compute individual contributions to each expert's likelihood as
\begin{align*}
    &L_{c,i} \gets Z_iT_i + (1-Z_i)(1-T_i)\\
    &L_{a,i} \gets T_i\\
    &L_{n,i} \gets 1-T_i\\
    &L_{d,i} \gets Z_i(1-T_i) + (1-Z_i)T_i
\end{align*}

\Iterate until convergence: \\
Compute the posterior probabilities associated with the nodes of the tree as \Comment{E-step}
\begin{align*}
    &h_{c,i} \gets g_{c,i}L_{c,i} / \sum_{s\in\{c,a,n,d\}} g_{s,i}L_{s,i} \\
    &h_{a,i} \gets g_{a,i}L_{a,i} / \sum_{s\in\{c,a,n,d\}} g_{s,i}L_{s,i} \\
    &h_{n,i} \gets g_{n,i}L_{n,i} / \sum_{s\in\{c,a,n,d\}} g_{s,i}L_{s,i} \\
    &h_{d,i} \gets g_{d,i}L_{d,i} / \sum_{s\in\{c,a,n,d\}} g_{s,i}L_{s,i} 
\end{align*}
For the gating network fit $\hat{\rho}_s(\cdot), s\in\{c,a,n,d\}$ as a multiclass classification \Comment{M-step}\\ 
$\quad$ problem with features $(X_i^{\hat{\rho}})_{1 \leq i \leq n}$, and targets $(h_{c,i}, h_{a,i}, h_{n,i}, h_{d,i})_{1 \leq i \leq n}.$ \\
Update the prior probabilities associated with the nodes of the tree as 
\begin{align*}
    g_{c,i} \gets \hat{\rho}_c(X_i^\rho), \quad 
    g_{a,i} \gets \hat{\rho}_a(X_i^\rho), \\
    g_{n,i} \gets \hat{\rho}_n(X_i^\rho), \quad 
    g_{d,i} \gets \hat{\rho}_d(X_i^\rho)
\end{align*}

\Endloop $\hat{\rho}_c(\cdot)$
\end{footnotesize}
\end{algorithmic}
\end{algorithm}

\begin{algorithm}[H]
\caption{The nonparametric EM-like procedure for estimating $Q_c(\cdot;\zeta)$ when $Y$ is binary. Here, $Q_c$ denotes either $Q_{c11}$ or $Q_{c00}$ depending on the data subset considered.}\label{algo:5}
\begin{algorithmic}
\begin{footnotesize}
\Input Data $(X_i^{\zeta}, Y_i,P_{cll}(X_i))_{i:Z_i=l,T_i=l}$ where $l\in\{0,1\}$ and $X_i^{\zeta}$ is a relevant subset of the variables contained in $X_i.$

\Initialize the prior probabilities associated with the nodes of the tree as
\begin{equation*}
g_{c,i} \gets P_{cll}(X_i), \quad
g_{nc,i} \gets 1 - P_{cll}(X_i).     
\end{equation*}

\Initialize the individual predictions from the expert networks $\big(Q_c(\cdot),Q_{nc}(\cdot)\big)$ as 
\begin{equation*}
Q_{c,i} \sim \mathcal{U}_{[0,1]} \quad
Q_{nc,i} \sim \mathcal{U}_{[0,1]}
\end{equation*}

\Iterate until convergence: \\
Compute individual contributions to each expert's likelihood as
\begin{align*}
    &L_{c,i} \gets Q_{c,i}^{Y_i} (1-Q_{c,i})^{1-Y_i}\\
    &L_{nc,i} \gets Q_{nc,i}^{Y_i} (1-Q_{nc,i})^{1-Y_i}
\end{align*}
Compute the posterior probabilities associated with the nodes of the tree as \Comment{E-step}
\begin{align*}
    &h_{c,i} \gets g_{c,i}L_{c,i} / \Big(g_{c,i}L_{c,i} + g_{nc,i}L_{nc,i} \Big) \\
    &h_{nc,i} \gets 1 - h_{nc,i}
\end{align*}
For the expert network $Q_c(\cdot)$ fit $\hat{Q}_c(\cdot)$ via a weighted nonparametric classifier \Comment{M-step}\\
$\quad$ with features $X_i^{\zeta}$, targets $Y_i$ and weights $h_{c,i}$. \\
For the expert network $Q_{nc}(\cdot)$ fit $\hat{Q}_{nc}(\cdot)$ via a weighted nonparametric classifier \\
$\quad$ with features $X_i^{\zeta}$, targets $Y_i$ and weights $h_{nc,i}$. \\
Update the predictions from the expert networks $\big(Q_c(\cdot),Q_{nc}(\cdot)\big)$ as 
\begin{equation*}
Q_{c,i} \gets \hat{Q}_c(X_i^\zeta)\quad
Q_{nc,i} \gets \hat{Q}_{nc}(X_i^\zeta)
\end{equation*}
\Endloop $\hat{Q}_{c}(\cdot)$
\end{footnotesize}
\end{algorithmic}
\end{algorithm}

\begin{algorithm}[H]
\caption{The nonparametric EM-like procedure for estimating $Q_c(\cdot)$ when Y is continuous. Here, $Q_c$ denotes either $Q_{c11}$ or $Q_{c00}$ depending on the data subset considered.}\label{algo:6}

\begin{algorithmic}
\begin{footnotesize}
\Input Data $(X_i^{\zeta}, Y_i,P_{cll}(X_i))_{i:Z_i=l,T_i=l}$ where $l\in\{0,1\}$ and $X_i^{\zeta}$ is a relevant subset of the variables contained in $X_i.$

\Initialize the prior probabilities associated with the nodes of the tree as
\begin{equation*}
g_{c,i} \gets P_{cll}(X_i), \quad
g_{nc,i} \gets 1 - P_{cll}(X_i).     
\end{equation*}

\Initialize the variance parameters $(\sigma_{c}^2,\sigma_{nc}^2)$
\begin{align*}
\sigma_c^2 \gets 1 \quad \sigma_{nc}^2 \gets 1 
\end{align*}

\Initialize the individual predictions from the expert networks $\big(Q_c(\cdot),Q_{nc}(\cdot)\big)$ as 
\begin{equation*}
Q_{c,i} \sim \mathcal{N}(0, 1) \quad
Q_{nc,i} \sim \mathcal{N}(0, 1)
\end{equation*}
\Iterate until convergence: \\
Compute individual contributions to each expert's likelihood as
\begin{align*}
    &L_{c,i} \gets \mathcal{N}_\mathcal{L}(Y_i|\mu=Q_{c,i},\sigma^2=\sigma_c^2) \\
    &L_{nc,i} \gets \mathcal{N}_\mathcal{L}(Y_i|\mu=Q_{nc,i},\sigma^2=\sigma_{nc}^2)
\end{align*}
Compute the posterior probabilities associated with the nodes of the tree as \Comment{E-step}
\begin{align*}
    &h_{c,i} \gets g_{c,i}L_{c,i} / \Big(g_{c,i}L_{c,i} + g_{nc,i}L_{nc,i} \Big) \\
    &h_{nc,i} \gets 1 - h_{nc,i}
\end{align*}
For the expert network $Q_c(\cdot)$ fit $\hat{Q}_c(\cdot)$ via weighted nonparametric regression \Comment{M-step}\\   
$\quad$ with features $X_i^{\zeta}$, targets $Y_i$ and weights $h_{c,i}$. \\
For the expert network $Q_{nc}(\cdot)$ fit $\hat{Q}_{nc}(\cdot)$ via weighted nonparametric regression \\   
$\quad$ with features $X_i^{\zeta}$, targets $Y_i$ and weights $h_{nc,i}$. \\
Update the variance parameter for each expert $\big(Q_c(\cdot),Q_{nc}(\cdot)\big)$ as
\begin{align*}
\sigma_c^2 &\gets  \frac{1}{n'}\sum_{i=1}^{n'} h_{c,i}\big( Y_i - \hat{Q}_{c}(X_i^{\zeta})\big)^2, \\
\sigma_{nc}^2 &\gets  \frac{1}{n'}\sum_{i=1}^{n'} h_{nc,i}\big( Y_i - \hat{Q}_{nc}(X_i^{\zeta})\big)^2
\end{align*}
Update the predictions from the expert networks $\big(Q_c(\cdot),Q_{nc}(\cdot)\big)$ as 
\begin{equation*}
Q_{c,i} \gets \hat{Q}_{c}(X_i^{\zeta}) \quad
Q_{nc,i} \gets \hat{Q}_{nc}(X_i^{\zeta})
\end{equation*}
\Endloop $\hat{Q}_{c}(\cdot)$
\end{footnotesize}
\end{algorithmic}
\end{algorithm}

\begin{algorithm}[H]
\caption{The EM procedure for estimating $Q_c(\cdot;\zeta),$ assuming exclusion restriction, when $Y$ is binary. Here, $Q_c$ denotes either $Q_{c11}$ or $Q_{c00}$ depending on the data subset considered.}\label{algo:7}
\begin{algorithmic}
\begin{footnotesize}
\Input Data $(X_i^{\zeta}, Y_i,P_{c\cdot m}(X_i),P_{k\cdot m}(X_i),P_{d\cdot m}(X_i))_{i:T_i=m}$ where $\{k,m\} \in \big\{\{a,1\},\{n,0\}\big\}$ and $X_i^{\zeta}$ is a relevant subset of the variables contained in $X_i.$

\Initialize the prior probabilities associated with the nodes of the tree as
\begin{equation*}
g_{c,i} \gets P_{c\cdot m}(X_i), \quad
g_{k,i} \gets P_{k\cdot m}(X_i), \quad    
g_{d,i} \gets P_{d\cdot m}(X_i).
\end{equation*}

\Initialize the parameters $(\zeta_c,\zeta_{nc})$ of the experts $\big(Q_c(\cdot),Q_{nc}(\cdot)\big)$ at random e.g., 
\begin{equation*}
\zeta_c \sim \mathcal{N}(0, D), \quad
\zeta_k \sim \mathcal{N}(0, D), \quad
\zeta_d \sim \mathcal{N}(0, D)
\end{equation*}
with $D$ a diagonal matrix.

\Compute individual predictions from the initiated expert networks $\big(Q_c(\cdot),Q_k(\cdot),Q_d(\cdot)\big)$ as 
\begin{align*}
&Q_{c,i} \gets \expit(\zeta_c^TX_i^\zeta), \\ 
&Q_{k,i} \gets \expit(\zeta_k^TX_i^\zeta), \\ 
&Q_{d,i} \gets \expit(\zeta_d^TX_i^\zeta).
\end{align*}

\Iterate until convergence on the parameters $\zeta=(\zeta_c^T,\zeta_k^T,\zeta_d^T)^T$: \\
Compute individual contributions to each expert's likelihood as
\begin{align*}
    &L_{c,i} \gets Q_{c,i}^{Y_i} (1-Q_{c,i})^{1-Y_i}\\
    &L_{k,i} \gets Q_{k,i}^{Y_i} (1-Q_{k,i})^{1-Y_i}\\
    &L_{d,i} \gets Q_{d,i}^{Y_i} (1-Q_{d,i})^{1-Y_i}&&
\end{align*}
Compute the posterior probabilities associated with the nodes of the tree as \Comment{E-step}
\begin{align*}
    &h_{c,i} \gets g_{c,i}L_{c,i} / \Big(g_{c,i}L_{c,i} + g_{k,i}L_{k,i} + g_{d,i}L_{d,i} \Big) \\
    &h_{k,i} \gets g_{k,i}L_{c,i} / \Big(g_{c,i}L_{c,i} + g_{k,i}L_{k,i} + g_{d,i}L_{d,i} \Big) \\
    &h_{d,i} \gets g_{d,i}L_{c,i} / \Big(g_{c,i}L_{c,i} + g_{k,i}L_{k,i} + g_{d,i}L_{d,i} \Big)
\end{align*}
For the expert network $Q_c(\cdot),Q_k(\cdot),Q_d(\cdot)$ estimate parameters $\zeta_c,\zeta_k,\zeta_d$ \\
by solving the IRLS problems \Comment{M-step}   
\begin{align*}
    \zeta_c \gets &\argmax\limits_{\zeta_c} \sum_{i=1}^{n'} h_{c,i} \Big[ Y_i\ln{\{ \expit(\zeta_c^TX_i^\zeta) \}}+ (1-Y_i)\ln{ \{1-\expit(\zeta_c^TX_i^\zeta) \} }\Big]\\
    \zeta_k \gets &\argmax\limits_{\zeta_k} \sum_{i=1}^{n'} h_{k,i} \Big[ Y_i\ln{\{ \expit(\zeta_k^TX_i^\zeta) \}}+ (1-Y_i)\ln{ \{1-\expit(\zeta_k^TX_i^\zeta) \} }\Big]\\
    \zeta_d \gets &\argmax\limits_{\zeta_d} \sum_{i=1}^{n'} h_{d,i} \Big[ Y_i\ln{\{ \expit(\zeta_d^TX_i^\zeta) \}}+ (1-Y_i)\ln{ \{1-\expit(\zeta_d^TX_i^\zeta) \} }\Big] 
\end{align*}
Update the predictions from the expert networks $\big(Q_c(\cdot),Q_{nc}(\cdot)\big)$ as 
\begin{align*}
&Q_{c,i} \gets  \expit(\zeta_c^TX_i^\zeta), \\
&Q_{k,i} \gets  \expit(\zeta_k^TX_i^\zeta), \\
&Q_{d,i} \gets  \expit(\zeta_d^TX_i^\zeta)&& 
\end{align*}
\Endloop $Q_c(x;\hat{\zeta})=\expit{(\zeta_c^Tx)}$
\end{footnotesize}
\end{algorithmic}
\end{algorithm}


\begin{algorithm}[H]
\caption{The EM procedure for estimating $Q_c(\cdot;\zeta),$ assuming exclusion restriction, when Y is continuous. Here, $Q_c$ denotes either $Q_{c11}$ or $Q_{c00}$ depending on the data subset considered.}\label{algo:8}
\begin{algorithmic}
\begin{footnotesize}
\Input Data $(X_i^{\zeta}, Y_i,P_{c\cdot m}(X_i),P_{k\cdot m}(X_i),P_{d\cdot m}(X_i))_{i:T_i=m}$ where $\{k,m\} \in \big\{\{a,1\},\{n,0\}\big\}$ and $X_i^{\zeta}$ is a relevant subset of the variables contained in $X_i.$

\Initialize prior probabilities associated with the nodes of the tree as
\begin{equation*}
g_{c,i} \gets P_{c\cdot m}(X_i), \quad
g_{k,i} \gets P_{k\cdot m}(X_i), \quad    
g_{d,i} \gets P_{d\cdot m}(X_i).
\end{equation*}
\Initialize the parameters $\Big((\zeta_c,\sigma_c^2),(\zeta_k,\sigma_k^2),(\zeta_d,\sigma_d^2)\Big)$ of the experts $\big(Q_c(\cdot),Q_k(\cdot),Q_d(\cdot)\big)$ at random e.g., 
\begin{align*}
\zeta_c \sim \mathcal{N}(0, D) \quad  \sigma_c^2 \gets 1,&
&\zeta_k \sim \mathcal{N}(0, D) \quad  \sigma_k^2 \gets 1, \\
\zeta_d \sim \mathcal{N}(0, D) \quad  \sigma_d^2 \gets 1,&& \text{with $D$ a diagonal matrix.}
\end{align*}

\Compute individual predictions from the initiated expert networks $\big(Q_c(\cdot),Q_k(\cdot),Q_d(\cdot)\big)$ as 
\begin{align*}
&Q_{c,i} \gets \zeta_c^TX_i^\zeta &
&Q_{k,i} \gets \zeta_k^TX_i^\zeta &
&Q_{d,i} \gets \zeta_d^TX_i^\zeta
\end{align*}
\Iterate until convergence on the parameters $\zeta=(\zeta_c^T,\zeta_k^T,\zeta_d^T)^T$: \\
Compute individual contributions to each expert's likelihood as
\begin{align*}
    L_{c,i} \gets \mathcal{N}_\mathcal{L}(Y_i|\mu=Q_{c,i},\sigma^2=\sigma_c^2)\\
    L_{k,i} \gets \mathcal{N}_\mathcal{L}(Y_i|\mu=Q_{k,i},\sigma^2=\sigma_k^2)\\
    L_{d,i} \gets \mathcal{N}_\mathcal{L}(Y_i|\mu=Q_{d,i},\sigma^2=\sigma_d^2)
\end{align*}
Compute the posterior probabilities associated with the nodes of the tree as \Comment{E-step}
\begin{align*}
    &h_{c,i} \gets g_{c,i}L_{c,i} / \Big(g_{c,i}L_{c,i} + g_{k,i}L_{k,i} + g_{d,i}L_{d,i} \Big) \\
    &h_{k,i} \gets g_{k,i}L_{c,i} / \Big(g_{c,i}L_{c,i} + g_{k,i}L_{k,i} + g_{d,i}L_{d,i} \Big) \\
    &h_{d,i} \gets g_{d,i}L_{c,i} / \Big(g_{c,i}L_{c,i} + g_{k,i}L_{k,i} + g_{d,i}L_{d,i} \Big)
\end{align*}
For the expert networks $\big(Q_c(\cdot),Q_k(\cdot),Q_d(\cdot)\big)$ estimate parameters $(\zeta_c, \zeta_k, \zeta_d)$ \\
$\quad$ by solving the weighted least-squares problem \Comment{M-step}   
\begin{align*}
    \zeta_c \gets \argmin\limits_{\zeta_c} \sum_{i=1}^{n'} h_{c,i} \big( Y_i - \zeta_c^TX_i^{\zeta}\big)^2 \\
    \zeta_k \gets \argmin\limits_{\zeta_k} \sum_{i=1}^{n'} h_{k,i} \big( Y_i - \zeta_k^TX_i^{\zeta}\big)^2 \\
    \zeta_d \gets \argmin\limits_{\zeta_d} \sum_{i=1}^{n'} h_{d,i} \big( Y_i - \zeta_d^TX_i^{\zeta}\big)^2 
\end{align*}
Update the variance parameter for each expert $\big(Q_c(\cdot),Q_k(\cdot),Q_d(\cdot)\big)$ as
\begin{align*}
\sigma_j^2 \gets & \frac{1}{n'-d}\sum_{i=1}^{n'} h_{j,i}\big( Y_i - \zeta_j^TX_i^{\zeta}\big)^2 
\end{align*}
with $j \in \{c, k,d\}$.\\
Update the predictions from the expert networks $\big(Q_c(\cdot),Q_k(\cdot),Q_d(\cdot)\big)$ as 
\begin{align*}
&Q_{c,i} \gets \zeta_c^TX_i^\zeta &
&Q_{k,i} \gets \zeta_k^TX_i^\zeta &
&Q_{d,i} \gets \zeta_d^TX_i^\zeta &
\end{align*}
\Endloop $Q_c(x;\hat{\zeta})=\zeta_c^Tx$
\end{footnotesize}
\end{algorithmic}
\end{algorithm}
\begin{algorithm}[H]
\caption{The nonparametric EM-like procedure for estimating $Q_c(\cdot),$ assuming exclusion restriction, when $Y$ is binary. Here, $Q_c$ denotes either $Q_{c11}$ or $Q_{c00}$ depending on the data subset considered.}\label{algo:9}
\begin{algorithmic}
\begin{footnotesize}
\Input Data $(X_i^{\zeta}, Y_i,P_{c\cdot m}(X_i),P_{k\cdot m}(X_i),P_{d\cdot m}(X_i))_{i:T_i=m}$ where $\{k,m\} \in \big\{\{a,1\},\{n,0\}\big\}$ and $X_i^{\zeta}$ is a relevant subset of the variables contained in $X_i.$

\Initialize the prior probabilities associated with the nodes of the tree as
\begin{equation*}
g_{c,i} \gets P_{c\cdot m}(X_i), \quad
g_{k,i} \gets P_{k\cdot m}(X_i), \quad    
g_{d,i} \gets P_{d\cdot m}(X_i).
\end{equation*}

\Initialize the individual predictions from the expert networks $\big(Q_c(\cdot),Q_k(\cdot),Q_d(\cdot)\big)$ as 
\begin{align*}
&Q_{c,i} \sim \mathcal{U}_{[0,1]} &
&Q_{k,i} \sim \mathcal{U}_{[0,1]} &
&Q_{d,i} \sim \mathcal{U}_{[0,1]} &
\end{align*}

\Iterate until convergence: \\
Compute individual contributions to each expert's likelihood as
\begin{align*}
    &L_{c,i} \gets Q_{c,i}^{Y_i} (1-Q_{c,i})^{1-Y_i}\\
    &L_{k,i} \gets Q_{k,i}^{Y_i} (1-Q_{k,i})^{1-Y_i}\\
    &L_{d,i} \gets Q_{d,i}^{Y_i} (1-Q_{d,i})^{1-Y_i}
\end{align*}
Compute the posterior probabilities associated with the nodes of the tree as \Comment{E-step}
\begin{align*}
    &h_{c,i} \gets g_{c,i}L_{c,i} / \Big(g_{c,i}L_{c,i} + g_{nc,i}L_{nc,i} \Big) \\
    &h_{k,i} \gets g_{k,i}L_{c,i} / \Big(g_{c,i}L_{c,i} + g_{nc,i}L_{nc,i} \Big) \\
    &h_{d,i} \gets g_{d,i}L_{c,i} / \Big(g_{c,i}L_{c,i} + g_{nc,i}L_{nc,i} \Big)
\end{align*}
For the expert network $Q_c(\cdot)$ fit $\hat{Q}_c(\cdot)$ via a weighted nonparametric classifier\Comment{M-step}\\
$\quad$ with features $X_i^{\zeta}$, targets $Y_i$ and weights $h_{c,i}$. \\
For the expert network $Q_k(\cdot)$ fit $\hat{Q}_k(\cdot)$ via a weighted nonparametric classifier \\
$\quad$ with features $X_i^{\zeta}$, targets $Y_i$ and weights $h_{k,i}$. \\
For the expert network $Q_d(\cdot)$ fit $\hat{Q}_d(\cdot)$ via a weighted nonparametric classifier \\
$\quad$ with features $X_i^{\zeta}$, targets $Y_i$ and weights $h_{d,i}$. \\
Update the predictions from the expert networks $\big(Q_c(\cdot),Q_k(\cdot),Q_d(\cdot)\big)$ as 
\begin{align*}
&Q_{c,i} \gets \hat{Q}_c(X_i^\zeta)&
&Q_{k,i} \gets \hat{Q}_k(X_i^\zeta)&
&Q_{d,i} \gets \hat{Q}_s(X_i^\zeta)&
\end{align*}
\Endloop $\hat{Q}_{c}(\cdot)$
\end{footnotesize}
\end{algorithmic}
\end{algorithm}

\begin{algorithm}[H]
\caption{The nonparametric EM-like procedure for estimating $Q_c(\cdot),$ assuming exclusion restriction, when $Y$ is continuous. Here, $Q_c$ denotes either $Q_{c11}$ or $Q_{c00}$ depending on the data subset considered.}\label{algo:10}
\begin{algorithmic}
\begin{footnotesize}
\Input Data $(X_i^{\zeta}, Y_i,P_{c\cdot m}(X_i),P_{k\cdot m}(X_i),P_{d\cdot m}(X_i))_{i:T_i=m}$ where $\{k,m\} \in \big\{\{a,1\},\{n,0\}\big\}$ and $X_i^{\zeta}$ is a relevant subset of the variables contained in $X_i.$

\Initialize the prior probabilities associated with the nodes of the tree as
\begin{equation*}
g_{c,i} \gets P_{c\cdot m}(X_i), \quad
g_{k,i} \gets P_{k\cdot m}(X_i), \quad    
g_{d,i} \gets P_{d\cdot m}(X_i).
\end{equation*}

\Initialize the variance parameters $(\sigma_{c}^2,\sigma_k^2,\sigma_d^2)$
\begin{align*}
& \sigma_c^2 \gets 1&
& \sigma_k^2 \gets 1&
& \sigma_d^2 \gets 1&
\end{align*}

\Initialize the individual predictions from the expert networks $\big(Q_c(\cdot),Q_k(\cdot),Q_d(\cdot)\big)$ as 
\begin{align*}
&Q_{c,i} \sim \mathcal{N}(0, 1) &
&Q_{k,i} \sim \mathcal{N}(0, 1) &
&Q_{d,i} \sim \mathcal{N}(0, 1) &
\end{align*}
\Iterate until convergence: \\
Compute individual contributions to each expert's likelihood as
\begin{align*}
    &L_{c,i} \gets \mathcal{N}_\mathcal{L}(Y_i|\mu=Q_{c,i},\sigma^2=\sigma_c^2) \\
    &L_{k,i} \gets \mathcal{N}_\mathcal{L}(Y_i|\mu=Q_{k,i},\sigma^2=\sigma_k^2) \\
    &L_{d,i} \gets \mathcal{N}_\mathcal{L}(Y_i|\mu=Q_{d,i},\sigma^2=\sigma_d^2)
\end{align*}
Compute the posterior probabilities associated with the nodes of the tree as \Comment{E-step}
\begin{align*}
    &h_{c,i} \gets g_{c,i}L_{c,i} / \Big(g_{c,i}L_{c,i} + g_{nc,i}L_{k,i} + g_{nc,i}L_{d,i} \Big) \\
    &h_{k,i} \gets g_{k,i}L_{c,i} / \Big(g_{c,i}L_{c,i} + g_{nc,i}L_{k,i} + g_{nc,i}L_{d,i} \Big) \\
    &h_{d,i} \gets g_{d,i}L_{c,i} / \Big(g_{c,i}L_{c,i} + g_{nc,i}L_{k,i} + g_{nc,i}L_{d,i} \Big)
\end{align*}
For the expert network $Q_c(\cdot)$ fit $\hat{Q}_c(\cdot)$ via weighted nonparametric regression \Comment{M-step}\\   
 $\quad$ with features $X_i^{\zeta}$, targets $Y_i$ and weights $h_{c,i}$. \\
For the expert network $Q_k(\cdot)$ fit $\hat{Q}_k(\cdot)$ via weighted nonparametric regression \\   
 $\quad$ with features $X_i^{\zeta}$, targets $Y_i$ and weights $h_{k,i}$. \\
For the expert network $Q_d(\cdot)$ fit $\hat{Q}_d(\cdot)$ via weighted nonparametric regression \\\   
 $\quad$ with features $X_i^{\zeta}$, targets $Y_i$ and weights $h_{d,i}$. \\ 
Update the variance parameter for each expert $Q_c(\cdot)$ fit $\hat{Q}_c(\cdot)$ as
\begin{align*}
\sigma_c^2 & \gets  \frac{1}{n'}\sum_{i=1}^{n'} h_{c,i}\big( Y_i - \hat{Q}_{c}(X_i^{\zeta})\big)^2,  \\
\sigma_{nc}^2 & \gets  \frac{1}{n'}\sum_{i=1}^{n'} h_{nc,i}\big( Y_i - \hat{Q}_{nc}(X_i^{\zeta})\big)^2
\end{align*}
Update the predictions from the expert networks $Q_c(\cdot)$ fit $\hat{Q}_c(\cdot)$ as 
\begin{align*}
&Q_{c,i} \gets \hat{Q}_{c}(X_i^{\zeta}) &
&Q_{k,i} \gets \hat{Q}_{k}(X_i^{\zeta}) &
&Q_{d,i} \gets \hat{Q}_{d}(X_i^{\zeta}) &
\end{align*}
\Endloop $\hat{Q}_{c}(\cdot)$
\end{footnotesize}
\end{algorithmic}
\end{algorithm}

\begin{algorithm}[H]
\caption{The EM procedure for estimating $\rho_s(\cdot;\delta_s)$ where $s \in \{c,a,n\},$ assuming monotonicity.}\label{algo:11}
\begin{algorithmic}
\begin{footnotesize}
\Input Data $(X_i^\rho, Z_i, T_i)_{1 \leq i \leq n}$ where $X_i^{\rho}$ is a relevant subset of the variables contained in $X_i$. 

\Initialize the prior probabilities associated with the nodes of the tree as
\begin{equation*}
g_{c,i} \gets 1/3, \quad g_{a,i} \gets 1/3, \quad \text{and} \quad g_{n,i} \gets 1/3.    
\end{equation*}

\Compute individual contributions to each expert's likelihood as
\begin{align*}
    &L_{c,i} \gets Z_iT_i + (1-Z_i)(1-T_i)\\
    &L_{a,i} \gets T_i\\
    &L_{n,i} \gets 1-T_i
\end{align*}

\Iterate until convergence on the parameters $\delta=(\delta_c^T,\delta_a^T,\delta_n^T)^T$: \\
Compute the posterior probabilities associated with the nodes of the tree as \Comment{E-step}
\begin{align*}
    &h_{c,i} \gets g_{c,i}L_{c,i} / \sum_{s\in\{c,a,n,d\}} g_{s,i}L_{s,i} \\
    &h_{a,i} \gets g_{a,i}L_{a,i} / \sum_{s\in\{c,a,n,d\}} g_{s,i}L_{s,i} \\
    &h_{n,i} \gets g_{n,i}L_{n,i} / \sum_{s\in\{c,a,n,d\}} g_{s,i}L_{s,i}
\end{align*}
For the gating network $\rho(\cdot)$ estimate parameters $\delta$ by solving the IRLS problem\Comment{M-step}\\ 
\begin{equation*}
    \delta \gets \argmax\limits_{\delta} \sum_{i=1}^n \sum_{s\in\{c,a,n\}} h_{s,i} \ln{ \Big (\frac{\exp{\delta_s^TX_i^\rho}}{\sum_{k\in\{c,a,n\}}\exp{\delta_k^TX_i^\rho} }\Big)} 
\end{equation*}
$\quad$ as a multinomial logistic regression with features $(X_i^{\rho})_{1 \leq i \leq n}$,\\
$\quad$ and targets $(h_{c,i}, h_{a,i}, h_{n,i})_{1 \leq i \leq n}.$ \\
Update the prior probabilities associated with the nodes of the tree as 
\begin{align*}
    g_{c,i} \gets \exp{\delta_c^TX_i^\rho} / \sum_{k\in\{c,a,n\}}\exp{\delta_k^TX_i^\rho} \\
    g_{a,i} \gets \exp{\delta_a^TX_i^\rho} / \sum_{k\in\{c,a,n\}}\exp{\delta_k^TX_i^\rho} \\
    g_{n,i} \gets \exp{\delta_n^TX_i^\rho} / \sum_{k\in\{c,a,n\}}\exp{\delta_k^TX_i^\rho} 
\end{align*}

\Endloop $\rho_s(x;\hat{\delta})=\exp{\delta_s^Tx} / \sum_{k\in\{c,a,n\}}\exp{\delta_k^Tx}.$
\end{footnotesize}
\end{algorithmic}
\end{algorithm}


\begin{algorithm}[H]
\caption{The nonparametric EM-like procedure for estimating $\rho_s(\cdot)$ where $s \in \{c,a,n\},$ assuming monotonicity. }\label{algo:12}
\begin{algorithmic}
\begin{footnotesize}
\Input Data $(X_i^{\rho}, Z_i, T_i)_{1 \leq i \leq n}$ where $X_i^{\rho}$ is a relevant subset of the variables contained in $X_i$. 

\Initialize the prior probabilities associated with the nodes of the tree as
\begin{equation*}
g_{c,i} \gets 1/3, \quad g_{a,i} \gets 1/3, \quad \text{and} \quad g_{n,i} \gets 1/3.    
\end{equation*}

\Compute individual contributions to each expert's likelihood as
\begin{align*}
    &L_{c,i} \gets Z_iT_i + (1-Z_i)(1-T_i)\\
    &L_{a,i} \gets T_i\\
    &L_{n,i} \gets 1-T_i
\end{align*}

\Iterate until convergence: \\
Compute the posterior probabilities associated with the nodes of the tree as \Comment{E-step}
\begin{align*}
    &h_{c,i} \gets g_{c,i}L_{c,i} / \sum_{s\in\{c,a,n\}} g_{s,i}L_{s,i} \\
    &h_{a,i} \gets g_{a,i}L_{a,i} / \sum_{s\in\{c,a,n\}} g_{s,i}L_{s,i} \\
    &h_{n,i} \gets g_{n,i}L_{n,i} / \sum_{s\in\{c,a,n\}} g_{s,i}L_{s,i} 
\end{align*}
For the gating network fit $\hat{\rho}_s(\cdot), s\in\{c,a,n\}$ as a multiclass classification \Comment{M-step}\\ 
$\quad$ problem with features $(X_i^{\hat{\rho}})_{1 \leq i \leq n}$, and targets $(h_{c,i}, h_{a,i}, h_{n,i})_{1 \leq i \leq n}.$ \\
Update the prior probabilities associated with the nodes of the tree as 
\begin{align*}
    g_{c,i} \gets \hat{\rho}_c(X_i^\rho), \quad 
    g_{a,i} \gets \hat{\rho}_a(X_i^\rho), \quad
    g_{n,i} \gets \hat{\rho}_n(X_i^\rho)
\end{align*}

\Endloop $\hat{\rho}_c(\cdot)$
\end{footnotesize}
\end{algorithmic}
\end{algorithm}

\section{Identifiabilty results}\label{app:ident}
\begin{theorem}\label{thm:identif_mix1}
Let $\mathcal{Z} = \{0,1\}$ denote the space of allocated treatment, and assume that the space of covariable $\mathcal{X}$ is an open subset of $\R^d.$ Suppose the functions $\pi(\cdot\,;\beta) \colon \mathcal{X} \times \mathcal{Z} \to \R$ are of the form
\begin{align*}
    \pi(x,z\,; \beta)&=z \rho_{c}(x;\beta)+\rho_{a}(x;\beta) +0 \times \rho_{n}(x;\beta)+(1-z) \rho_{d}(x;\beta)
\end{align*}
where
\begin{align*}
    &\rho_k(x;\beta)=\frac{\exp(\beta^{T}_k x)}{1+ \sum_{l \in \{c,a,n\}} \exp(\beta^{T }_l x) }, \; \textnormal{for}~k \in \{ c, a, n\}, \\
&\rho_d(x;\beta)=\frac{1}{1+ \sum_{l \in \{c,a,n\}} \exp(\beta^{ T }_l x) }.
\end{align*}
with parameters $\beta=(\beta_c,\beta_a,\beta_n) \in \Theta=\R^d\setminus \{0\} \times \R^d \times \R^d.$
Then, the statistical model $\{\pi(\cdot;\beta)\}_{\beta \in \Theta}$ is identifiable.
\end{theorem}

\begin{remark}\label{rmk:identifiability}
In the above model, we assumed that $\beta_c \neq 0$. In view of our application, this assumption that $\beta_c \neq 0$ is very mild. Indeed, letting $\beta_c=0$ implies that $\rho_c=\rho_d$. In plain words, this means that, for every individual, the probability of being a complier and the probability of being a defier are the same. Thus, our assumption holds as soon as a single individual has a probability of being a complier that is different from 
    their probability of being a defier.
\end{remark}

\begin{proof}[Proof of Theorem \ref{thm:identif_mix1}]
 Consider  $\pi(\cdot;\beta)$ and $\pi(\cdot ;\gamma)$ with $\beta$ and $\gamma \in \Theta$. Assume that $\pi(\cdot;\beta)=\pi(\cdot;\gamma)$ on $\mathcal{X}\times \mathcal{Z}$.  Specializing in $z=0$ and $z=1$, we get the following equations on $\mathcal{X}$:

\begin{equation}\label{eq:master_sys}
    \begin{cases}
        \rho_c(x;\beta)+\rho_a(x;\beta)=\rho_c(x;\gamma)+\rho_a(x;\gamma)\\
        \rho_a(x;\beta)+\rho_d(x;\beta)=\rho_a(x;\gamma)+\rho_d(x;\gamma)  .
    \end{cases}
\end{equation}
The system \eqref{eq:master_sys} is equivalent to
\begin{equation}\label{eq:funda_sys}
    \begin{cases}
        \Big(1+ \sum \limits_{l \in \{c,a,n\}} \exp(\gamma^{ T }_l x)\Big) \sum \limits_{l \in \{c,a\}}\exp(\beta^{ T}_l x) =  \Big(1+ \sum \limits_{l \in \{c,a,n\}} \exp(\beta^{ T }_l x)\Big) \sum \limits_{l \in \{c,a\}} \exp(\gamma^{T}_l x)\\
        \Big(1+ \sum \limits_{l \in \{c,a,n\}} \exp(\gamma^{ T }_l x)\Big) \Big(1+ \exp(\beta^{ T}_a x) \Big)=  \Big(1+ \sum \limits_{l \in \{c,a,n\}} \exp(\beta^{ T }_l x)\Big) \Big(1+ \exp(\gamma^{T}_a x) \Big)
    \end{cases}
\end{equation}

Using the expansion of $\exp$ in power series, the first order term of the expansion of the System \eqref{eq:funda_sys} provides the following identities

\begin{equation*}
\begin{cases}
\beta_c+\beta_a+\gamma_n=\gamma_c+\gamma_a+\beta_n\\
\gamma_c+\gamma_n+\beta_a=\beta_c+\beta_n+\gamma_a
\end{cases}
\end{equation*}
which leads to
\begin{equation}\label{eq:id_step1}
\begin{cases}
\beta_c=\gamma_c\\
\gamma_n+\beta_a=\beta_n+\gamma_a.
\end{cases}
\end{equation}
Using again the expansion of $\exp$ in power series, the second order term of the expansion of the System \eqref{eq:funda_sys} provides the following identities
\begin{align*}
\begin{cases}
[\gamma_c^T x + \beta_a^T x]^2+[\gamma_n^T x + \beta_a^T x ]^2 + [\gamma_c^T x]^2+[\gamma_n^T x]^2\\=[\beta_c^T x + \gamma_a^T x]^2+[\beta_n^T x + \gamma_a^T x ]^2 + [\beta_c^T x]^2+[\beta_n^T x]^2\\
\\
[\gamma_c^T x + \beta_a^T x]^2+[\beta_n^T x + \gamma_c^T x ]^2 + [\gamma_c^T x]^2+[\gamma_a^T x]^2\\
=[\beta_c^T x + \gamma_a^T x]^2+[\beta_c^T x + \gamma_n^T x ]^2 + [\beta_c^T x]^2+[\beta_a^T x]^2.
\end{cases}
\end{align*}
Using the relations provided by the System \eqref{eq:id_step1}, the above equations simplify to

\begin{align*}
\begin{cases}
   [\gamma_c^T x + \beta_a^T x]^2 + [\gamma_n^T x]^2=[\beta_c^T x + \gamma_a^T x]^2+[\beta_n^T x]^2\\
   [\beta_n^T x + \gamma_c^T x ]^2+[\gamma_a^T x]^2=[\beta_c^T x + \gamma_n^T x ]^2+[\beta_a^T x]^2 .
\end{cases}
\end{align*}
Expanding and symplifying the above expression, we get
\begin{align*}
\begin{cases}
     [\beta_a^T x]^2+ 2 [\gamma_c^T x] [\beta_a^T x] + [\gamma_n^T x]^2
     =[\beta_n^T x]^2+2 [\beta_c^T x]  [\gamma_a^T x]+ [\gamma_a^T x]^2\\
   [\beta_a^T x]^2+2 [\beta_c^T x] [\gamma_n^T x ]+[\gamma_n^T x ]^2
   =[\gamma_a^T x]^2+2[\beta_n^T x] [\gamma_c^T x] + [\beta_n^T x]^2.
\end{cases}
\end{align*}
Substracting the second equation to the first one in the above system provides the relation
\begin{equation*}
    2 [\beta_c^T x] \left([\beta_a^T x]-[\gamma_n^T x ] \right)= 2 [\beta_c^T x] \left([\gamma_a^T x]-[\beta_n^T x ] \right).
\end{equation*}
That is
\begin{equation*}
    [\beta_c^T x] \left([\beta_a-\gamma_n-\gamma_a+\beta_n]^T x\right)= 0.
\end{equation*}
Consider the function $\phi \colon x \mapsto \beta_c^T x$ and $\psi \colon x \mapsto [\beta_a-\gamma_n-\gamma_a+\beta_n]^T x$. Assuming $\beta_c \neq 0$, we get that $\psi = 0$ on the open set $\mathcal{X} \setminus \ker \phi$. Since $\psi$ is linear, it is zero on all of $\R^d$. It follows that
\begin{equation*}
    \beta_a-\gamma_n=\gamma_a-\beta_n.
\end{equation*}
This relation together with the System of equations \eqref{eq:id_step1} implies that
\begin{align*}
   \beta_c=\gamma_c, \quad \beta_a=\gamma_a ,\quad \beta_n =\gamma_n. 
\end{align*}
\end{proof}
\begin{theorem}\label{thm:identif_mix2_lm}
Assume that the space of covariable $\mathcal{X}$ is a non-empty open subset of $\R^d$  containing zero and that the conditional probability function $P_{cll} \colon \mathcal{X} \to (0,1)$ is a non-constant continuous function. Then, the statistical models of the form
\begin{align*}
    \Big\{P_{cll}(x)\alpha^T x +\{1-P_{cll}(x)\}\beta^T x \colon (\alpha, \beta) \in \Theta \Big\}
\end{align*}
are identifiable.
\end{theorem}
%
\begin{proof}
Let $\{(\alpha, \beta),(\alpha', \beta') \}\in \Theta^2$ such that, $\forall x\in \mathcal{X}$  
\begin{align*}
P_{cll}(x)\alpha^T x +\{1-P_{cll}(x)\}\beta^T x= P_{cll}(x)\alpha'^T x +\{1-P_{cll}(x)\}\beta'^T x.
\end{align*}
From the last equation, algebra yields: $\forall x\in \mathcal{X},$
\begin{align*}
P_{cll}(x)(\alpha-\alpha')^T x +\{1-P_{cll}(x)\}(\beta-\beta')^T x = 0 
\end{align*}
\begin{align}\label{eq:identifiability_step2}
\Big[P_{cll}(x)\big((\alpha-\alpha')^T-(\beta-\beta')^T\big) +(\beta-\beta')^T \Big] x= 0
\end{align}

\begin{align}\label{eq:identiability}
P_{cll}(x)\Big[(\alpha-\alpha')^T - (\beta-\beta')^T \Big] x= -(\beta-\beta')^T x.
\end{align}
Since $P_{cll}$ does not vanish and $\mathcal{X}$ is an open set containing zero, it follows from Equation \eqref{eq:identiability} that the linear forms $ x \mapsto \Big[(\alpha-\alpha')^T - (\beta-\beta')^T \Big] x$ and $ x \mapsto -(\beta-\beta')^T x$ have the same kernel. This implies that there exists $\lambda \in \R$ such that 
\begin{equation*}
    (\alpha-\alpha') - (\beta-\beta')=\lambda (\beta-\beta')
\end{equation*}
Using Equation \eqref{eq:identifiability_step2}, and the above observation, we get
\begin{equation*}
  (1+ \lambda P_{cll}(x)) (\beta-\beta')^T x= 0 \quad \forall x \in \mathcal{X}
\end{equation*}
Since $P_{cll}$ is non-constant, there exists $z \in \mathcal{X}$ such that $1+ \lambda P_{cll}(z) \neq 0$.  Since $P_{cll}$ is continuous, there is an open neighborhood $B_z$ of $z$, such that 
\begin{equation*}
\forall x \in B_z, \quad 1+ \lambda P_{cll}(x) \neq 0.
\end{equation*}
Hence, for every $x \in B_z$, $(\beta-\beta')^T x=0$. Since, the linear form $x \mapsto (\beta-\beta')^T x$ vanishes on an open set, it is zero. This implies that $\beta=\beta^\prime$. In turn, this yields that $\alpha=\alpha^\prime$.
\end{proof}
\begin{theorem}\label{thm:identif_mix2_expit}
Let $\mathcal{X}$ be an open subset of $\R^d$ containing zero and define $ g_{\alpha,\beta} \colon \mathcal{X} \to \R$ as
\begin{equation*}
    g_{\alpha,\beta}(x) \! = \! P_{cll}(x)\expit{(\alpha^T x)} + (1-P_{cll}(x))\expit{(\beta^T x)}.
\end{equation*}
Assume that the function $P_{cll} \colon \mathcal{X} \to (0;1)$ is $\mathcal{C}^2$ and that $\frac{\partial P_{cll}(x)}{\partial x_i}(0)\neq 0$ for all $i\in\{1,2,\dots,d\}.$ 
Then, the statistical models of the form $\Big\{g_{\alpha,\beta} \colon (\alpha, \beta) \in \Theta \Big\}$
are identifiable. 
\end{theorem}

\begin{proof}
Let $\{(\alpha, \beta),(\alpha', \beta') \}\in \Theta^2$ such that, $\forall x\in \mathcal{X}$
\begin{align*}
p(x)\expit{(\alpha^T x)} +\{1-p(x)\}\expit{(\beta^T x)}
= p(x)\expit{(\alpha'^T x)} +\{1-p(x)\}\expit{(\beta'^T x)}
\end{align*}
Setting $p_i(t)=p(t e_i)$ where $e_i$ is the $i^{th}$ vector of the canonical basis of $\R^d$, the above expression yields for $i=1,\ldots,d:$
\begin{align} \label{eq:constraint_along_i}
p_i(t)\expit{(\alpha_i t)} +\{1-p_i(t)\}\expit{(\beta_i t)} 
= p_i(t)\expit{(\alpha_i t)} +\{1-p_i(t)\}\expit{(\beta_i t)}.
\end{align}

Consider, now, the Taylor expansions at order two in zero of the left and right hand-side of Equation \eqref{eq:constraint_along_i}, we get: 
\begin{align*}
\begin{cases}
   p(0)\dfrac{\alpha_i}{4}+ \{1-p(0)\}\dfrac{\beta_i}{4}=p(0)\dfrac{\alpha_i'}{4}+\{1-p(0)\}\dfrac{\beta_i'}{4}\\
  \dfrac{\alpha_i}{4}\der{p_i}{}(0)-\dfrac{\beta_i}{4} \der{p_i}{}(0)=\dfrac{\alpha_i'}{4} \der{p_i}{}(0)-\dfrac{\beta_i'}{4}\der{p_i}{}(0).
\end{cases}
\end{align*}

Since $\der{p_i}{}(0)=\dfrac{\partial p}{\partial x_i}(0)\neq 0$ for all $i\in\{1,2,\dots,d\},$ this system further simplifies to:
\begin{align*}
\begin{cases}
   p(0)\{\alpha_i-\beta_i\}+\beta_i=p(0)\{\alpha'_i-\beta'_i\}+\beta'_i\\
   \alpha_i-\beta_i=\alpha_i'-\beta_i'.
\end{cases}
\end{align*}
Substitution of the second row into the first one yields:
\begin{align*}
\begin{cases}
   \beta_i=\beta'_i\\
   \alpha_i=\alpha_i'.
\end{cases}
\end{align*}
for every $i\in\{1,2,\dots,d\}$. That is, $(\alpha, \beta)=(\alpha', \beta').$
\end{proof}

\section{Asymptotic properties}\label{app:asymp}
Assuming parametric generalized linear models for $\rho_k(\cdot)$, $k\in\{c,a,n,d\}$ as well as  $Q_{c11}(\cdot)$, $Q_{a11}(\cdot)$ and $Q_{c00}(\cdot)\,Q_{n00}(\cdot),$ the estimator $\widehat{\Delta}_{PI}$ can be expressed as 
\begin{equation}\label{eq:estimator}
    \widehat{\Delta}_{PI}=\frac{\sum_{i=1}^n \big\{ Q_{c11}(X_i;\hat{\beta})-Q_{c00}(X_i;\hat{\gamma})\big\}\rho_c(X_i;\hat{\delta})}{\sum_{i=1}^n\rho_c(X_i;\hat{\delta})}.
\end{equation}
Rearranging Equation \eqref{eq:estimator}, we note that $\widehat{\Delta}_{PI}$ is the solution of an equation of the form
\begin{equation*}
    \sum_{i=1}^n \psi_0(X_i;\Delta,\delta,\beta,\gamma)=0
\end{equation*}
where
\begin{align*}
    \psi_0(x;\Delta_{PI},\delta,\beta,\gamma) &= \big\{ Q_{c11}(x;\beta)-Q_{c00}(x;\gamma) -\Delta_{PI}\big\}\rho_c(x;\delta).
\end{align*}
We note that $\hat{\beta}=(\hat{\beta}_c^T,\hat{\beta}_{nc}^T)^T$, $\hat{\gamma}=(\hat{\gamma}_c^T,\hat{\gamma}_{nc}^T)^T$ and $\hat{\delta}=(\hat{\delta}_c^T,\hat{\delta}_a^T,\hat{\delta}_n^T,\hat{\delta}_d^T)^T$ are parameters from mixture of expert models. In fact, each parameter $\hat{\beta},\hat{\gamma},\hat{\delta}$ solve an estimating (score) equation from the corresponding mixture of expert model. For example, the parameters $\hat{\delta}$ solve an equation of the form
\begin{align*}
     &\sum_{i=1}^n \psi_1(X_i,Z_i,T_i;\delta) = 0 
\end{align*}
    where
\begin{align*}
\psi_1(x,z,t;\delta)&= \frac{\partial }{\partial \delta} \ln p_{T|X,Z}(t|x,z;\delta),
\end{align*}
and 
\[
p_{T|X,Z}(t|x,z;\delta)= \hspace{-0.5cm} \sum_{s \in \{c,a,n,d\}}  \hspace{-0.3cm}\frac{\exp (\delta_s^T x) \,\mu_s(z)^{t}(1-\mu_s(z))^{1-t}}{\sum_{k \in \{c,a,n,d \}} \exp \delta_k^T x}.
\]
 In addition, the parameters $\hat{\beta},\hat{\gamma}$ solve equations of the form
\begin{align*}
     &\sum_{i=1}^n \psi_2(X_i,Z_i,T_i,Y_i;\delta,\beta) = 0, \\
     &\sum_{i=1}^n \psi_3(X_i,Z_i,T_i,Y_i;\delta,\gamma) = 0 
\end{align*}
where
\begin{align*}
\psi_2(x,z,t,y;\delta,\beta)&= \mathbbm{1}(z=1)\mathbbm{1}(t=1) \frac{\partial }{\partial \beta}\ln \Big\{ P_{c11}(x;\delta) p_{Y|X}(y|x;\beta_c)+ P_{a11}(x;\delta) p_{Y|X}(y|x;\beta_{nc}) \Big\},\\
\psi_3(x,z,t,y;\delta,\gamma)&= \mathbbm{1}(z=0)\mathbbm{1}(t=0) \frac{\partial }{\partial \gamma} \ln \Big\{ P_{c00}(x;\delta) p_{Y|X}(y|x;\gamma_c)+ P_{n00}(x;\delta) p_{Y|X}(y|x;\gamma_{nc}) \Big\},
\end{align*}
and
\begin{equation*}
p_{Y|X}(y|x;\theta)=\expit(\theta ^T x )^{y} (1-\expit(\theta ^T x ))^{1-y}
\end{equation*}
if $Y$ is binary, or denoting $\theta=(\theta_\zeta^T,\theta_\sigma)^T,$
\begin{equation*}
p_{Y|X}(y|x;\theta)= \frac{1}{{\theta_\sigma \sqrt{2\pi}}} \exp\left\{-\frac{(y - \theta_\zeta^T x)^2}{2\theta_\sigma^2}\right\}
\end{equation*}
if $Y$ is continuous.
The representations above allow to define the following estimating function 
\begin{align*}
\psi(x,z,t,y;\Delta_{PI},\delta,\beta,\gamma)=&\Big(\psi_0^T(x;\Delta_{PI},\delta,\beta,\gamma), \psi_1^T(x,z,t;\delta), \psi_2^T(x,z,t,y;\delta,\beta), \\
&\psi_3^T(x,z,t,y;\delta,\gamma)\Big)^T.
\end{align*}
It can be shown that for all $(\Delta,\delta,\beta,\gamma),$
\begin{equation*}
    \EX_{\Delta,\delta,\beta,\gamma}\Big[\psi(X,Z,T,Y;\Delta,\delta,\beta,\gamma)\Big]=0.
\end{equation*}
Thus, $\psi(x,z,t,y;\Delta,\delta,\beta,\gamma)$ is an unbiased estimating function and $\widehat{\Delta}_{PI}$ is a partial M-estimator of $\psi-$type. Applying standard results from M-estimation theory (\emph{see} for example Equation (7.10) from \citet[p. 301]{stefanski2002}), $(\widehat{\Delta}_{PI},\hat{\delta}^T,\hat{\beta}^T,\hat{\gamma}^T)^T$ are consistent and asymptotically normal estimators for $(\Delta,\delta^T,\beta^T,\gamma^T)^T$; that is, denoting true values of the parameters with subscript $0$,
\begin{equation*}
    (\widehat{\Delta}_{PI},\hat{\delta}^T,\hat{\beta}^T,\hat{\gamma}^T)^T \quad \xrightarrow{p} \quad (\Delta_0,\delta_0^T,\beta_0^T,\gamma_0^T)^T
\end{equation*}
and
 \begin{equation*}
     n^{1/2}         
             \begin{pmatrix}
           \widehat{\Delta}_{PI} - \Delta_0\\ 
          \hat{\delta} - \delta_0 \\
           \hat{\beta} - \beta_0\\
           \hat{\gamma} - \gamma_0
    \end{pmatrix}
         \xrightarrow{\mathcal{D}} \mathcal{N}(0,\Sigma),
 \end{equation*}
where the variance-covariance matrix $\Sigma$ is given by the sandwich formula at the point $v_0=(\Delta_0,\delta_0,\beta_0,\gamma_0)$: 
\begin{align*}
\Sigma &= A^{-1}B\{A^{-1}\}^T \quad \text{with}   \\ 
    A&= \EX_{v_0} \bigg[\frac{\partial \psi(X,Z,T,Y;\Delta,\delta,\beta,\gamma) }{\partial (\Delta,\delta^T,\beta^T,\gamma^T)}\bigg\rvert_{v_0}\bigg], \\
    B&= \EX_{v_0} \bigg[ \psi(X,Z,T,Y;v_0)\, \psi^T(X,Z,T,Y;v_0)\bigg].
\end{align*}
In principle, a closed-form estimator for the asymptotic variance of $\widehat{\Delta}_{PI}$ could be derived by calculating the top-left element of the matrix $\Sigma$. However, because the estimator $\widehat{\Delta}_{PI}$ involve iterative fitting of three mixture of expert models, carrying out the required derivations would be a formidable task. Derivations could be conducted numerically via the \texttt{R} package \texttt{geex} \citep{geex2020} or algorithmically and symbolically through the \texttt{Mestim} package \citep{mestim2022}. However, in our experience both these methods appeared slow and numerically unstable. Accordingly, we recommend that measures of uncertainty for $\widehat{\Delta}_{PI}$ be obtained via a standard nonparametric bootstrap.

\section{Simulation}
\subsection{Description}\label{app:sim_desc}
In this section, we provide further descriptions of the data-generating mechanism used in the simulations. We generate synthetic datasets comprising $7$ Bernoulli and $7$ log-normally distributed, correlated, covariates $X=(X_{1}, X_{2},\dots,X_{14})$ as follows.

\begin{enumerate}
\item We randomly generate correlated intermediate covariates $X_{1}', X_{2}',\dots,X_{14}'$ from a multivariate gaussian distribution
    \begin{displaymath}
    \left(X_{1}', X_{2}', \ldots,  X_{14}' \right)^{T} \sim \mathcal{N}(0,\,\Sigma).
    \end{displaymath}
To generate $\Sigma$, we chose $14$ eigenvalues $(\lambda_{i})_{1  \leq i \leq 14}$ with
$\lambda_i= 1+ (i-1) \times 0.2$, and sample a random orthogonal matrix $O$ of size $14 \times 14$. The covariance matrix $\Sigma$ is obtained via
  \begin{align*}
     \Sigma &=O
     \begin{bmatrix} 
    \lambda_{1} & 0               & \dots  & 0\\
    0           & \lambda_{2}     & \ddots & \vdots\\
    \vdots      & \ddots          & \ddots & 0\\
    0           & \dots           & 0      & \lambda_{14}
    \end{bmatrix}
    O^{T}.
  \end{align*}
\item To allow for the Bernoulli or log-normal distribution of covariates, we generate $$X_{1}, X_{2},\dots,X_{14}$$ as follows 
$$ (X_{1},\dots,X_{7}) =(\mathbbm{1}\{X_{1}'>0\},\dots,\mathbbm{1}\{X_{7}'>0\}), $$
$$(X_{8},\dots,X_{14}) =\big(\exp(X_{8}'),\dots,\exp(X_{14}')\big).$$ We add $X_{0} \equiv 1$ to allow for intercepts.
\item We generate data from the covariates in this manner.\\
The strata are
\begin{align*}
    &\boldsymbol{S}|X \sim \textit{Multinomial}\Big(N=1,p=\big(\rho_c(X),\rho_a(X),\rho_n(X),\rho_d(X)\big)^T\Big)
\end{align*}
with
\begin{align*}
    &\forall s \in \{c,a,n,d\}, \quad \rho_s(X)=\frac{\exp{\delta_s^TX}} { \sum_{k\in\{c,a,n,d\}}\exp{\delta_k^TX}}
\end{align*}
and the parameters $\delta=(\delta_c^T,\delta_a^T,\delta_n^T,\delta_d^T)^T$ are set at random: $\delta \sim \mathcal{U}[-1,1]^{4\times15}.$
The allocated treatment is $Z \sim \textit{Bernoulli}(0.5)$. The treatment effectively taken is $T=S_c Z + S_a + S_d  (1-Z)$. For every $k \in \{c,a,n,d\}$, $l \in \{0,1\}$ and $m \in \{0,1\}$, the elementary potential outcomes are generated as
\begin{align*}
    Y^{s=k,z=l,t=m}|X \sim \textit{Bernoulli}(\expit \beta_{klm}^TX),
\end{align*}
where the parameters $\beta_{klm}$ are set at random: $\beta_{klm} \sim \mathcal{U}[-1,1]^{15}$, for every $k \in \{c,a,n,d\}$, $l \in \{0,1\}$, $m \in \{0,1\}$.
The potential outcomes are 
\begin{align*}
Y^{t=1}&= S_cY^{s=c,z=1,t=1} + S_aZY^{s=a,z=1,t=1} + S_a(1-Z)Y^{s=a,z=0,t=1} + S_dY^{s=d,z=0,t=1}, \\
Y^{t=0}&= S_cY^{s=c,z=0,t=0} + S_nZY^{s=n,z=1,t=0} + S_n(1-Z)Y^{s=n,z=0,t=0} + S_dY^{s=d,z=1,t=0},
\end{align*}
while observed outcomes are $Y=TY^{t=1}+(1-T)Y^{t=0}.$
\end{enumerate}
In the scenarios where the exclusion restriction assumption holds, we set $Y^{s=a,z=1,t=1} \gets Y^{s=a,z=0,t=1}$ and $Y^{s=n,z=1,t=0} \gets Y^{s=n,z=0,t=0}.$ When the monotonicity assumption holds we set $\delta_d^T X \gets -\infty$, so that $\rho_d(\cdot) \equiv 0.$ For well specified scenarios, all parametric models use covariates $(X_{0},X_{1},\dots,X_{14})$ as predictor variables. For misspecified scenarios, all parametric models use covariates $(X_{0},X_{1},\dots,X_{6})$ and $(X_{8},X_{9}\dots,X_{13})$ as predictor variables, such that a Bernoulli distributed and a log-normally distributed relevant variables are omitted.

\subsection{Instrumental variable methods}\label{app:sim_iv}
The Wald and IV matching estimators were calculated as follows \citep{wald1940,angrist1996,frolich2007}:
\begin{align*}
    &\widehat{\Delta}_{IV wald}=
    \frac{\dfrac{\sum_{i=1}^n Z_iY_i}{\sum_{i=1}^n Z_i}-\dfrac{\sum_{i=1}^n \{1-Z_i\}Y_i}{\sum_{i=1}^n \{1-Z_i\}}}{\dfrac{\sum_{i=1}^n Z_iT_i}{\sum_{i=1}^n Z_i}-\dfrac{\sum_{i=1}^n \{1-Z_i\}T_i}{\sum_{i=1}^n \{1-Z_i\}}} \\
    &\\
    &\widehat{\Delta}_{IV\,matching}=\frac{\sum_{i\in\{0,1\},j} (-1)^{i+1} \hat{\EX}[Y|G=j,Z=i]}{\sum_{i \in\{0,1\},j} (-1)^{i+1} \hat{\EX}[T|G=j,Z=i]}.
\end{align*}
In the equation above, $\{Z=1\}$ and $\{Z=0\}$ observations are matched on $\hat{\eta}(X)$, and $G\in\mathbb{N}$ denotes the group an observation belongs to; $\hat{\EX}[\cdot]$ are calculated as (weighted) averages.
\clearpage
\subsection{Supplementary results}\label{app:sim_supp_res}
\begin{figure}[h!]
\centering
\scalebox{.525}{
\includegraphics{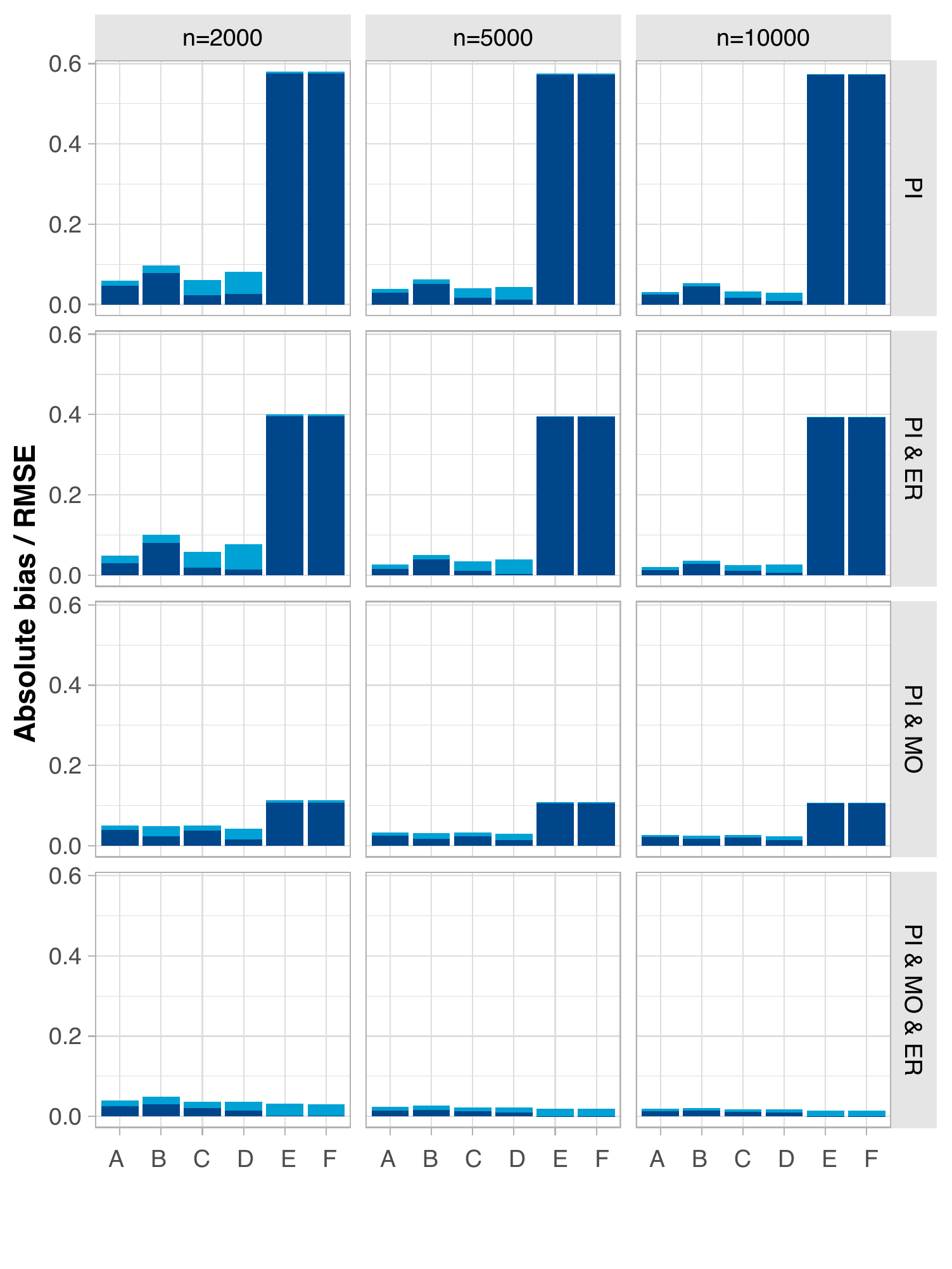}}
\caption{Estimators’ absolute bias and Root Mean Squared Error (RMSE) with parametric models including all the relevant variables, across twelve scenario/sample size combinations.}
\scriptsize  Absolute bias is the darker portion of each bar; RMSE corresponds to the total bar size. Letters A, B, C, D, E and F indicate the estimators $\widehat{\Delta}_{PI}$, $\widehat{\Delta}_{PI}^{ER}$, $\widehat{\Delta}_{PI,MO}$, $\widehat{\Delta}_{PI,MO}^{ER}$, $\widehat{\Delta}_{IV matching}$, and $\widehat{\Delta}_{IV wald}$ respectively. Abbreviations: PI = Principal Ignorability (Scenario 1), PI \& ER = Principal Ignorability and Exclusion Restriction (Scenario 2), PI \& MO= Principal Ignorability and Monotonicity (Scenario 3), PI \& ER \& MO = Principal Ignorability, Exclusion Restriction and Monotonicity (Scenario 4).
\end{figure}

\begin{figure}
\centering
\scalebox{.55}{
\includegraphics{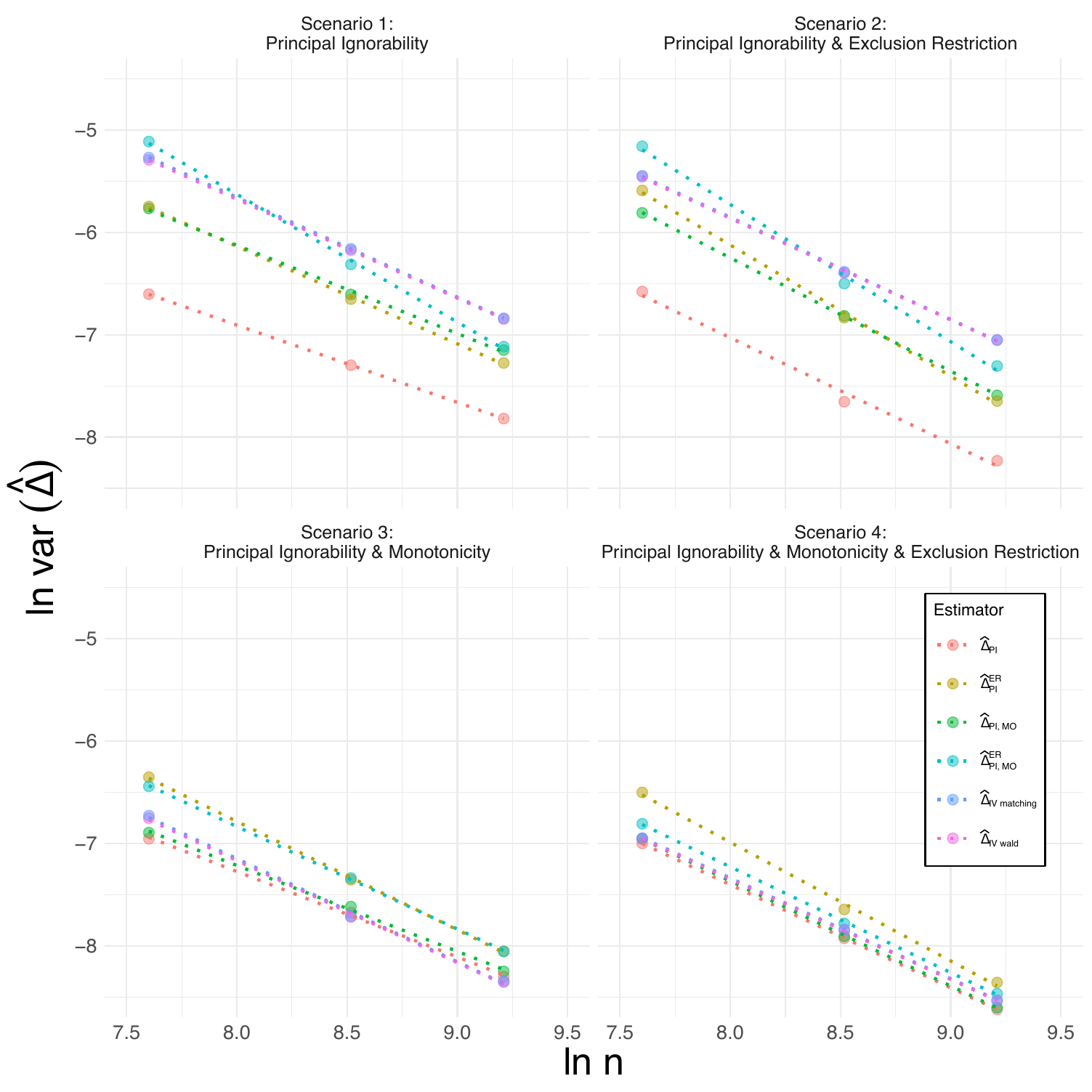}}
\caption{Estimators' variance and rate of convergence with parametric models including all the relevant variables.}
\scriptsize For each estimator/scenario combination, slopes describe rates of convergence (e.g., a slope of $-1/2$ points to a convergence speed of $\sqrt{n}$), while intercepts approximate the logarithm of asymptotic variances.
\end{figure}

\begin{table*}
\caption{Results of the simulation study where parametric models include all relevant variables.\label{tab:sim_res_appendix}}
\centering
\scalebox{.55}{
\begin{tabular}{lcccccccccccc}
\hline
 Assumptions  &   \multicolumn{3}{c}{Scenario 1}  &  \multicolumn{3}{c}{Scenario 2}  &  \multicolumn{3}{c}{Scenario 3}  & \multicolumn{3}{c}{Scenario 4} \\ 
Principal ignorability &  \multicolumn{3}{c}{$+$}  &  \multicolumn{3}{c}{$+$}  &  \multicolumn{3}{c}{$+$}  &  \multicolumn{3}{c}{$+$}  \\
Exclusion restriction &  \multicolumn{3}{c}{$-$}  &  \multicolumn{3}{c}{$+$}  &  \multicolumn{3}{c}{$-$}  &  \multicolumn{3}{c}{$+$}  \\
Monotonicity &  \multicolumn{3}{c}{$-$}  &  \multicolumn{3}{c}{$-$}  &  \multicolumn{3}{c}{$+$}  &  \multicolumn{3}{c}{$+$}  \\ \hline
Estimator & Bias (\%) & SE (\%) & RMSE  (\%) & Bias (\%) & SE (\%) & RMSE  (\%) & Bias (\%) & SE (\%) & RMSE  (\%) & Bias (\%) & SE (\%) & RMSE  (\%) \\ \hline 
$n=2\,000$ & & & & & & & & & & & &  \\ 
$\widehat{\Delta}_{PI}$ & $-$4.65 & {\bf 3.68} & {\bf 5.93} & $-$2.99 & {\bf 3.73} & {\bf 4.78} & $-$3.95 &  {\bf 3.09} &  5.02 & $-$2.50 & {\bf 3.02} & 3.92 \\ 
$\widehat{\Delta}_{PI}^{ER}$ & $-$7.85 & 5.65 & 9.67 & $-$7.95 & 6.11 & 10.03 & $-$2.39 & 4.17 & 4.81 & $-$3.05 & 3.87 & 4.93 \\ 
$\widehat{\Delta}_{PI,MO}$ & {\bf $-$2.33} & 5.59 & 6.06 & $-$1.90 & 5.48 & 5.79 & $-$3.85 & 3.18 & 4.99 & $-$2.08 & 3.09 & 3.72 \\ 
$\widehat{\Delta}_{PI,MO}^{ER}$ & $-$2.68 & 7.76 & 8.20 & {\bf $-$1.46} & 7.58 & 7.72 & {\bf $-$1.64} & 3.99 & {\bf 4.32} & $-$1.48 & 3.32 & 3.64 \\ 
$\widehat{\Delta}_{IV\,\text{matching}}$ & 57.57 &  7.17 & 58.01 & 39.55 & 6.56 & 40.09 & 10.78 & 3.46 & 11.32 & 0.13 & 3.10 & 3.10 \\ 
$\widehat{\Delta}_{IV\,\text{Wald}}$ & 57.52 & 7.09 & 57.96 & 39.53 & 6.53 & 40.07 & 10.76 & 3.41 & 11.29 & {\bf 0.12} & 3.08 & {\bf 3.09} \\ 

$n=5\,000$ & & & & & & & & & &  \\ 
$\widehat{\Delta}_{PI}$ & -2.92 & {\bf 2.60} & {\bf 3.91} & $-$1.59 & {\bf 2.18} & {\bf 2.69} & $-$2.52 & 2.15 & 3.32 & -1.52 & {\bf 1.90} & 2.43\\ 
$\widehat{\Delta}_{PI}^{ER}$ & $-$5.09 & 3.59 & 6.23 & $-$3.87 & 3.28 & 5.07 & $-$1.76 & 2.53 & 3.08 & $-$1.62 & 2.19 & 2.72 \\ 
$\widehat{\Delta}_{PI,MO}$ & $-$1.61 & 3.68 & 4.01 & $-$1.01 & 3.31 & 3.46 & $-$2.38 & 2.22 & 3.25 & $-$1.23 &  1.92 & 2.28 \\ 
$\widehat{\Delta}_{PI,MO}^{ER}$ & {\bf $-$1.22} & 4.26 & 4.43 & {\bf 0.25} & 3.88 & 3.88 & {\bf $-$1.46} & 2.55 & {\bf 2.94} & $-$1.00 & 2.04 & 2.27 \\ 
$\widehat{\Delta}_{IV\,\text{matching}}$ & 57.29 & 4.59 & 57.47 & 39.36 & 4.11 & 39.57 & 10.64 & 2.12 & 10.84 & {\bf 0.02} & 1.98 & 1.98 \\ 
$\widehat{\Delta}_{IV\,\text{Wald}}$ & 57.28 & 4.56 & 57.46 & 39.36 & 4.08 & 39.57 & 10.64 & {\bf 2.11} & 10.85 & 0.03 & 1.97 &  {\bf 1.97}  \\ 

$n=10\,000$ & & & & & & & & & & & &  \\ 
$\widehat{\Delta}_{PI}$ & $-$2.40 & {\bf 2.00} & 3.13 & $-$1.24 & {\bf 1.63} & {\bf 2.05} & $-$2.19 & 1.57 & 2.69 & $-$1.34 & {\bf 1.34} & 1.90 \\ 
$\widehat{\Delta}_{PI}^{ER}$ & $-$4.58 & 2.63 & 5.29 & $-$2.87 & 2.18 & 3.60 & $-$1.69 & 1.78 & 2.46 & $-$1.42 & 1.53 & 2.09 \\ 
$\widehat{\Delta}_{PI,MO}$ & $-$1.69 & 2.80 & 3.27 & $-$1.07 & 2.25 & 2.49 & $-$2.10 & 1.62 & 2.65 & $-$1.11 & 1.35 & 1.75 \\ 
$\widehat{\Delta}_{PI,MO}^{ER}$ &{\bf $-$0.90} & 2.85 & {\bf 2.99} &  {\bf 0.62} & 2.59 & 2.66 &  {\bf 1.44} & 1.78 &  {\bf 2.29} & $-$0.90 & 1.45 & 1.70 \\ 
$\widehat{\Delta}_{IV\,\text{matching}}$ & 57.21 &  3.27 & 57.31 & 39.25 & 2.94 & 39.36 & 10.58 & 1.54 & 10.70  & $-$0.02 & 1.40 & 1.40 \\ 
$\widehat{\Delta}_{IV\,\text{Wald}}$ & 57.22 & 3.26 & 57.32 & 39.26 & 2.94 & 39.37 & 10.59& {\bf 1.53} & 10.70 & {\bf $-$0.02} & 1.40 & {\bf 1.40} \\ 
\hline
\end{tabular}
}
{\small {\raggedright Scenario 1, $\EX[Y^{t=1}-Y^{t=0}]= -2.13\% $,  $\Delta= 20.31\% $, $\EX[(Y^{s=a,z=0,t=1}-Y^{s=a,z=1,t=1})^2]= 60.74\% $, $\EX[(Y^{s=n,z=1,t=0}-Y^{s=a,z=0,t=0})^2]= 54.90\% $, $\EX[\rho_c(X)]= 55.08\% $, and $\EX[\rho_d(X)]= 14.36\% $.  \par}
{\raggedright Scenario 2, $\EX[Y^{t=1}-Y^{t=0}]= -5.79\% $,  $\Delta= 20.31\% $, $\EX[(Y^{s=a,z=0,t=1}-Y^{s=a,z=1,t=1})^2]= 0 $, $\EX[(Y^{s=n,z=1,t=0}-Y^{s=a,z=0,t=0})^2]= 0 $, $\EX[\rho_c(X)]= 55.08\% $, and $\EX[\rho_d(X)]= 14.36\% $.  \par}
{\raggedright Scenario 3, $\EX[Y^{t=1}-Y^{t=0}]= 13.54\% $,  $\Delta= 20.31\% $, $\EX[(Y^{s=a,z=0,t=1}-Y^{s=a,z=1,t=1})^2]= 60.74\% $, $\EX[(Y^{s=n,z=1,t=0}-Y^{s=a,z=0,t=0})^2]= 54.90\% $, $\EX[\rho_c(X)]= 65.30\% $, and $\EX[\rho_d(X)]=0 $. \par}
{\raggedright  Scenario 4, $\EX[Y^{t=1}-Y^{t=0}]= 10.08\% $,  $\Delta= 20.31\% $, $\EX[(Y^{s=a,z=0,t=1}-Y^{s=a,z=1,t=1})^2]= 0$, $\EX[(Y^{s=n,z=1,t=0}-Y^{s=a,z=0,t=0})^2]= 0 $, $\EX[\rho_c(X)]= 56.30\% $, and $\EX[\rho_d(X)]= 0 $.  \par}}
\end{table*}

\vskip 0.2in
\clearpage
\bibliography{main}

\end{document}